\documentclass[a4paper,UKenglish,cleveref, autoref, thm-restate, numberwithinsect]{lipics-v2021}
\title{Cops and Robbers for Graphs on Surfaces with Crossings} %TODO Please add

\author{Prosenjit Bose}{School of Computer Science, Carleton University, Ottawa, Canada}{jit@scs.carleton.ca}{https://orcid.org/0000-0002-8906-0573}{}

\author{Pat Morin}{School of Computer Science, Carleton University, Ottawa, Canada}{morin@scs.carleton.ca}{https://orcid.org/0000-0003-0471-4118}{}

\author{Karthik Murali}{School of Computer Science, Carleton University, Ottawa, Canada}{KarthikMurali@cmail.carleton.ca}{https://orcid.org/0009-0003-3985-3609}{}%TODO mandatory, please use full name; only 1 author per \author macro; first two parameters are mandatory, other parameters can be empty. Please provide at least the name of the affiliation and the country. The full address is optional. Use additional curly braces to indicate the correct name splitting when the last name consists of multiple name parts.

\authorrunning{P. Bose, P. Morin and K. Murali} %TODO mandatory. First: Use abbreviated first/middle names. Second (only in severe cases): Use first author plus 'et al.'

\Copyright{Prosenjit Bose, Pat Morin and Karthik Murali} %TODO mandatory, please use full first names. LIPIcs license is "CC-BY";  http://creativecommons.org/licenses/by/3.0/

 %TODO mandatory: Please choose ACM 2012 classifications from https://dl.acm.org/ccs/ccs_flat.cfm 

 %TODO mandatory; please add comma-separated list of keywords

\category{} %optional, e.g. invited paper

\relatedversion{} %optional, e.g. full version hosted on arXiv, HAL, or other respository/website
\funding{Research supported in part by NSERC.}

\nolinenumbers %uncomment to disable line numbering

\hideLIPIcs

\usepackage{mathtools}
\usepackage{multicol}
\usepackage{cleveref}
\usepackage{makecell}
\usepackage[table]{xcolor}
\usepackage{lipsum}
\usepackage{float}

\newcommand{\edit}[1]{{\color{blue} #1}}
\renewcommand{\edit}[1]{#1}

\usepackage{apptools}
\AtAppendix{\counterwithin{lemma}{section}}

\DeclarePairedDelimiter{\ceil}{\lceil}{\rceil}
\DeclarePairedDelimiter{\floor}{\lfloor}{\rfloor}

\newtheorem{problem}{Open Problem}

\begin{document}

\maketitle

%TODO mandatory: add short abstract of the document
\begin{abstract}
\textit{Cops and Robbers} is a game played on a graph where a set of cops attempt to capture a single robber. The game proceeds in rounds, where each round first consists of the cops' turn, followed by the robber's turn. \edit{In the first round, the cops place themselves on a subset of vertices, after which the robber chooses a vertex to place himself. From the next round onwards,} in the cops' turn, every cop can choose to either stay on the same vertex or move to an adjacent vertex, and likewise the robber in his turn. The robber is considered to be captured if, at any point in time, there is some cop on the same vertex as the robber. The cops win if they can capture the robber within a finite number of rounds; else the robber wins. 

A natural question in this game concerns the \textit{cop-number} of a graph---the minimum number of cops needed to capture a robber. It has long been known that graphs embeddable (without crossings) on surfaces of bounded genus have bounded cop-number. In contrast, it was shown recently that the class of 1-planar graphs---graphs that can be drawn on the plane with at most one crossing per edge---does not have bounded cop-number. This paper initiates an investigation into how the distance between crossing pairs of edges influences a graph's cop number. In particular, we look at \textit{Distance $d$ Cops and Robbers}, a variant of the classical game, where the robber is considered to be captured if there is a cop within distance $d$ of the robber.

Let $c_d(G)$ denote the minimum number of cops required in the graph $G$ to capture a robber within distance $d$. We look at various classes of graphs, such as 1-plane graphs, $k$-plane graphs (graphs where each edge is crossed at most $k$ times), and even general graph drawings, and show that if every crossing pair of edges can be connected by a path of small length, then $c_d(G)$ is bounded, for small values of $d$. \edit{For example, we show that if a graph $G$ admits a drawing in which every pair of crossing edges is contained in a path of length 3, then $c_4(G) \leq 21$.} And if the drawing permits a stronger assumption that the endpoints of every crossing induce the complete graph $K_4$, then $c_3(G) \leq 9$. The tools and techniques that we develop in this paper are sufficiently general, enabling us to examine graphs drawn not only on the sphere but also on orientable and non-orientable surfaces.
\end{abstract}

\newpage
% \tableofcontents

\newpage
\setcounter{page}{1}

\section{Introduction}

\textit{Pursuit-evasion} is a general mathematical framework for problems that involve a set of \textit{pursuers} attempting to capture a set of \textit{evaders}. Such problems have many applications in robotics \cite{application_robotics}, network security \cite{application_network_security}, surveillance \cite{application_surveillance}, etc. The game of \textit{Cops and Robbers}, introduced independently by Quilliot \cite{quilliot1978}, and Nowakowski and Winkler \cite{nowakowski_winkler}, belongs to the family of pursuit-evasion problems where a set of cops (the pursuers) attempt to capture a single robber (the evader) on a graph. The game is played in rounds, with each round consisting first of the \textit{cops' turn}, followed by the \textit{robber's turn}. Initially, the cops first position themselves on a set of vertices, after which the robber chooses a vertex to place himself. In subsequent rounds, the cops' turn consists of each cop either staying on the same vertex or moving to a neighbouring vertex, and likewise the robber in his turn. The game ends when the robber is \textit{captured} by a cop; this happens when the robber is on the same vertex as some cop. If the robber has a strategy to permanently evade the cops, then the robber wins; else the cops win. (See the book \cite{bonato_book} by Bonato and Nowakowski for an extensive introduction to Cops and Robbers.)

One of the central questions in Cops and Robbers is to identify graph classes that require few cops to capture the robber. Formally, the \textit{cop-number} of a graph $G$, denoted by $c(G)$, is defined as the minimum number of cops required to capture the robber. A graph class $\mathcal{G}$ is said to be \textit{cop-bounded} if $c(G) \leq p$ for some integer $p$ and all graphs $G \in \mathcal{G}$; else $\mathcal{G}$ is \textit{cop-unbounded}. Many graph classes are known to be cop-bounded; for example, chordal graphs \cite{nowakowski_winkler, quilliot1978}, planar graphs \cite{AignerFromme}, graphs of bounded genus \cite{quilliot1985_genus}, $H$-minor-free graphs \cite{Andreae_minor} and $H$-(subgraph)-free graphs \cite{joret}. On the other hand, examples of graph classes that are cop-unbounded include bipartite graphs \cite{bonato_book}, $d$-regular graphs, for all 
$d \geq 3$ \cite{Andreae_regular}, 1-planar graphs \cite{optimal1plane}, and graphs of diameter 2 \cite{distance_cops_robbers}. One of the deepest open problems on cop-number is Meyniel's conjecture which states that $c(G) \in O(\sqrt{n})$ for all $n$-vertex graphs $G$ (see \cite{meyniel_survey} for a survey paper on Meyniel's conjecture).

\subsection{Survey of Related Results.}\label{sec: survey}

In this paper, we are interested in the cop-number of graphs drawn on (orientable and non-orientable) surfaces with crossings. Most results in this area are concerned with cop-number of graphs embedded without crossings on surfaces with genus $g$. For planar graphs, which are graphs embeddable on the sphere ($g = 0$), Aigner and Fromme \cite{AignerFromme} proved that 3 cops are sufficient, and sometimes necessary. For graphs with (orientable) genus $g > 0$, the best known bound on the cop-number is $\frac{4}{3}g + \frac{10}{3}$ \cite{genus_currentbest}. A long-standing conjecture by Schroeder \cite{schroder2001copnumber} is that $c(G) \leq g + 3$ for all values of $g$. For non-orientable surfaces of genus $g$, Andreae \cite{Andreae_minor} gave the first result that $c(G) \in O(g)$, which was later improved by Nowakowski and Schr\"oder (in an unpublished work \cite{nowakowski1997bounding}) to $c(G) \leq 2g+1$. The current best result is by Clarke et al. \cite{joret_non_orientable} who show that $c(G)$ is at most the maximum of all cop-numbers of graphs embeddable on an orientable surface of genus $g-1$. (Refer to \cite{Topological_directions} for a survey of results, conjectures and open problems for cop-numbers of graphs on surfaces.) 

Ever since the appearance of the proof that $c(G) \leq 3$ for planar graphs \cite{AignerFromme}, it has been known that the geometrical representation of a graph plays an important role in bounding its cop-number. For instance, the strategy for guarding planar graphs was adapted to show that $c(G) \leq 9$ for unit disk graphs \cite{unit_disk}, and $c(G) \leq 15$ for string graphs \cite{string}. More recently, the study of cop-numbers for the beyond-planar graph classes was initiated by Durocher et al. \cite{optimal1plane}. They show that unlike planar graphs, the class of \textit{1-planar graphs}---which allow for one crossing per edge---is cop-unbounded. On the positive side, they show 3 cops are sufficient, and sometimes necessary, for a \textit{maximal 1-planar graph}: a 1-planar graph to which no edge can be added without violating 1-planarity (while staying simple). To prove this, they crucially use the fact that the endpoints of every crossing induce the complete graph $K_4$. In a recent paper, Bose et al. \cite{tcs_1planar} showed that even under the relaxed requirement that no crossing of a 1-plane graph is an \textit{$\times$-crossing}---a crossing whose endpoints induce a matching---the cop-number remains bounded, and is at most 21. As a corollary, their result implies that 1-planar graphs with large cop-numbers necessarily embed with a large number of $\times$-crossings. 

\begin{table}[ht!]
\begin{center}
\renewcommand{\arraystretch}{1.4}
\begin{tabular}{|m{0.11\linewidth}|m{0.25\linewidth}|m{0.25\linewidth}|m{0.31\linewidth}|}
 \hline
 \textbf{Category} & \textbf{Graph Class} & \textbf{Cop Number} & \textbf{Remarks} \\
 \hline
 \hline
 \multirow{4}{2cm}{1-Plane Graphs} 
 & Full 1-plane graphs & $c(G) \leq 3$ & \cref{cor: 1-planar,cor: map graphs}; generalizes result in \cite{optimal1plane} \\
 \cline{2-4}
 & 1-plane graphs without $\times$-crossings & $c(G) \leq 15$ & \cref{cor: 1-planar}; improves $c(G) \leq 21$ from \cite{tcs_1planar} \\
 \cline{2-4}
 & 1-plane graphs & \makecell[l]{$c_{X - 1}(G) \leq 6X - 3$ \\ $c_{x - 1}(G) \leq 6x + 9$} & \cref{cor: 1-planar} \\
 \cline{2-4}
 & Graphs with at most one crossing per face of $\mathrm{sk}(G)$ & Cop-unbounded & \cref{thm: single crossing per face}; holds even with distant crossings \\
 \hline
 \multirow{4}{2cm}{$k$-Plane Graphs} 
 & Full $k$-plane graphs & $c_2(G) \leq 12k-3$ & \cref{cor: k-planar} \\
 \cline{2-4}
 & $k$-plane graphs without $\times$-crossings & $c_3(G) \leq 18k-3$ & \cref{cor: k-planar} \\
 \cline{2-4}
 & $k$-plane graphs & \makecell[l]{$c_{X - 1}(G) \leq 6k(X + 1)-3$ \\ $c_{x - 1}(G) \leq 6k(x + 2) - 3$} & \cref{cor: k-planar} \\
 \cline{2-4}
 & $k$-Framed graphs & $c_{\left \lceil \frac{k+8}{3} \right \rceil}(G) \leq 21$ & \cref{cor: d-framed} \\
 \hline
 \multirow{6}{2cm}{General Graphs} 
 & Graphs where all crossings are full & $c_3(G) \leq 9$ & \cref{cor: arbitrary x and X = 1} \\
 \cline{2-4}
 & Graphs without $\times$-crossings & $c_4(G) \leq 21$ & \cref{cor: arbitrary x and X = 1} \\
 \cline{2-4}
 & Map graphs & $c(G) \leq 3$ & \cref{thm: cop number map graph}; includes full 1-plane and optimal 2-plane graphs \\
 \cline{2-4}
 & All graphs & \makecell[l]{$c_{\left \lceil \frac{3}{2}(X+1) \right \rceil}(G) \leq 9$ \\ $c_{\left \lceil \frac{3}{2}(x+2) \right \rceil}(G) \leq 15$} & \cref{cor: estimates} \\
 \cline{2-4}
 & All graphs, $1 < \alpha < 3/2$ & \makecell[l]{$c_{\left \lceil \alpha(X+1) \right \rceil}(G) \leq \frac{8}{\alpha - 1}$ \\ $c_{\left \lceil \alpha(x+2) \right \rceil}(G) \leq \frac{11}{\alpha - 1}$} & \cref{cor: estimates} \\
 \hline
 \hline
 &  &  & For all integers $\mu \geq 6$, the set $\{c_{\lfloor X/6\rfloor -1}(G): X(G) \geq \mu\}$ is unbounded (\cref{thm: x/6 cop number})
 \\
 \hline
\end{tabular}
\caption{Summary of results for various graph classes drawn on the sphere with crossings.}
\label{table:results}
\end{center}
\end{table}

\subsection{Our Contribution.}

Motivated by \cite{tcs_1planar}, we undertake a comprehensive study of how the distance between crossing pairs of edges affects the cop number of graphs. Unlike \cite{tcs_1planar}, our aim is to go beyond 1-planar graphs, and look at general graphs drawn on surfaces, under the restriction that every crossing pair of edges can be connected by a path of small length. (A similar approach was undertaken in \cite{biedl2024parameterized} to generalise results on vertex connectivity from 1-plane graphs without $\times$-crossings to arbitrary graphs drawn on the sphere.) We consider two notions of distance between crossing pairs of edges. For any pair of edges $\{e_1,e_2\}$ that cross in $G$, we say that two endpoints of the crossing are \textit{consecutive} if one endpoint is a vertex incident with $e_1$ and the other is a vertex incident with $e_2$. For any fixed drawing of a graph $G$ on a surface, let $x(G)$ (resp. $X(G)$) be the smallest integer such that for every crossing in the drawing, there is a path of length at most $x(G)$ (resp. $X(G)$) connecting some (resp. every) pair of consecutive endpoints of the crossing. By the definition above, $x(G) \leq X(G)$; furthermore, $X(G) \leq x(G)+2$ since there exists a path of length at most $x(G)+2$ at every crossing where the first and last edges of the path are precisely the crossing pair of edges. Notwithstanding the tight relation between $X(G)$ and $x(G)$, we differentiate between the two parameters to capture their varying impact on cop-number. Indeed, one may already notice this for 1-plane graphs. If $G$ is a maximal 1-plane graph, where the endpoints of every crossing induce a $K_4$, we have $X(G) = 1$ and $c(G) \leq 3$ \cite{optimal1plane}. On the other hand, if $G$ is a 1-plane graph without $\times$-crossings, we have $x(G) = 1$ and $c(G) \leq 21$ \cite{tcs_1planar}. 

Most results that we present in this paper concern \textit{Distance $d$ Cops and Robbers}---a variant of the classical Cops and Robbers in which the robber is considered to be captured if, at any point in time, there is a cop within distance $d$ of the robber \cite{distance_cops_robbers}. We use the notation $c_d(G)$ to denote the minimum number of cops required to capture the robber within distance $d$ in the graph $G$. In this paper, we give several interesting results on distance $d$ cop-numbers of graphs $G$, where $d$ is parameterised by $X(G)$ and $x(G)$. We restrict our discussion here only to graphs drawn on the sphere, even though they extend to graphs drawn on surfaces with crossings. Our results not only encompass and improve the existing results on 1-plane graphs \cite{optimal1plane,tcs_1planar}, but also greatly simplify the existing proofs. The tools that we develop in this paper are equipped to handle beyond 1-planar graphs, such as $k$-plane graphs (where each edge is allowed to cross at most $k$ times), and even arbitrary graphs on surfaces. 

\subsubsection{Our Results.}

We show that graphs that embed on the plane with small values of $X(G)$ or $x(G)$ have small cop-numbers. (A summary of our results is given in \cref{table:results}.) Let us call a crossing in a graph $G$ a \textit{full crossing} if its endpoints induce the complete graph $K_4$, and an \textit{$\times$-crossing} if its endpoints induce a matching. We show that for any graph where all crossings are full, $c_3(G) \leq 9$, and if no crossing of the graph is an $\times$-crossing, then $c_4(G) \leq 21$. For larger values of $X(G)$ or $x(G)$, we show that $c_{\ceil{\alpha(X+1)}}(G)$ and $c_{\ceil{\alpha(x+2)}}(G)$ are in $O(\frac{1}{\alpha-1})$, for any $1 < \alpha < 3/2$. For values of $\alpha \geq 3/2$, the cop-numbers are at most 9 and 15, respectively. In contrast to this, we show that for every integer $\mu \geq 6$, $\max\{c_{\floor{X/6}-1}(G): X(G) \geq \mu\}$ is unbounded. We also consider the question of standard cop-numbers $c(G)$ for special classes of graphs. Let $\mathrm{sk}(G)$ denote the subgraph induced by the set of uncrossed edges of $G$. We show that for the class of \textit{map graphs}, which are precisely graphs that have an embedding on the plane such that all crossed edges are inserted as cliques inside some faces of $\mathrm{sk}(G)$, the cop-number is at most 3. (In fact, this generalises the result $c(G) \leq 3$ for maximal 1-plane graphs \cite{optimal1plane}.) On the other hand, for graphs where there is at most one crossing in each face of $\mathrm{sk}(G)$, we show that $c(G)$ is unbounded, even when no pair of crossings are close to each other. However, if one were to restrict the number of edges that bound each face of $\mathrm{sk}(G)$ by an integer $k$, then we are in the class of \textit{$k$-framed graphs}, and we show that $c_{\ceil{\frac{k+8}{3}}}(G) \leq 21$.

\subsubsection{Organisation.}

\cref{sec: prelims} provides the necessary preliminaries and formal definitions used throughout the paper. In \cref{sec: distance d cops and robber}, we investigate the distance-$d$ variant of Cops and Robbers and introduce two fundamental operations that simplify the problem on surfaces by reducing it to the study of 1-plane drawings on surfaces. \cref{sec: 1-planar} focuses on 1-planar graphs drawn on the sphere, and explores the consequences of our result for $k$-plane graphs and general graphs drawn on the sphere. \cref{sec: graphs on surfaces} extends our study to graphs drawn on surfaces with crossings, covering both orientable and non-orientable surfaces, and also discussing map graphs on surfaces. Finally, \cref{sec: conclusion} outlines a set of open problems and future research directions.

\section{Preliminaries} \label{sec: prelims}
 
All graphs in this paper are finite and undirected. If $H$ is a subgraph of $G$, then we write $H \subseteq G$. For any distinct pair of vertices $s$ and $t$ in $G$, an $(s,t)$-path is a simple path in $G$ with $s$ and $t$ as its end vertices. If a path $P$ contains vertices $u$ and $v$, then the subpath of $P$ containing all vertices from $u$ to $v$ will be denoted by $P(u \dots v)$. The length of a path $P$ is the number of edges on $P$, and is denoted by $|P|$. If $G$ is edge-weighted and $P$ is a path in $G$, then $|P|$ denotes the sum of weights of all edges in $P$.

\subparagraph{Graph Drawings.} A \textit{drawing} of a graph $G = (V,E)$ on a surface is a mapping of its vertex set $V(G)$ to distinct points on the surface and a mapping of each edge $e = (u,v)$ to a non-self-intersecting curve that has $u$ and $v$ as its endpoints and does not pass through any other vertex of $G$. In all drawings that we consider, any pair of edges can intersect only a finite number of times, no three edges are allowed to intersect at a point that is not a vertex, and no two edges can touch each other tangentially. (This is not a restriction because one can re-draw edges within a local neighbourhood of points on the surface to satisfy the requirements.) A pair of distinct edges are said to \textit{cross} each other if there is a point on the surface interior to both edges. For any graph $G$ drawn on a surface $S$, the \textit{faces} of $G$ are the connected regions of $S \setminus G$. The \textit{boundary of a face} is a sequence of edges or part-edges (the interior of an edge with at least one end being a crossing point) that bound the face. The \textit{skeleton} of a graph $G$ drawn on some surface is the subgraph $\mathrm{sk}(G)$ induced by all uncrossed edges of $G$. Hence, any drawing of a graph $G$ can be viewed as the union of $\mathrm{sk}(G)$ together with the set of crossed edges drawn inside the faces of $\mathrm{sk}(G)$. 

\subparagraph{Crossings and $(\Sigma,1)$-Planar Graphs.} A graph is \textit{planar} if it has a drawing on the sphere such that no two edges cross each other. If $G$ is a graph drawn on a surface $\Sigma$, then the \textit{$\Sigma$-planarisation} of $G$ is the graph $G^\times$ obtained by inserting a \textit{dummy vertex} at each crossing point in the drawing. Let $G$ be any graph drawn on a surface. Let $e_1$ and $e_2$ be two edges of $G$ that cross each other. The endpoints of $e_1$ and $e_2$ are together called \textit{endpoints of the crossing}. Two endpoints (possibly non-distinct) $v_1,v_2$ of the crossing are said to be \textit{consecutive} if $v_1$ is incident with $e_1$ and $v_2$ with $e_2$. Let $(u,v)$ and $(w,x)$ be a pair of edges that cross in a drawing of a graph $G$. Assume that the crossing has four distinct endpoints, and consider the graph $G[\{u,v,w,x\}]$ induced by these endpoints, up to counting parallel edges as a single edge. If $G[\{u,v,w,x\}]$ is isomorphic to the complete graph $K_4$, then we say that the crossing is \textit{full}. If $G[\{u,v,w,x\}]$ consists only of edges $(u,v)$ and $(w,x)$, then we say that the crossing is an \textit{$\times$-crossing}. (The nomenclature of crossings based on the subgraphs induced by its endpoints is borrowed from \cite{biedl_murali}.) 

A graph $G$ is a \textit{$(\Sigma,k)$-planar graph} if it has a drawing on a surface $\Sigma$ such that each edge is crossed at most $k$ times; a graph $G$ with such a drawing is called a \textit{$(\Sigma,k)$-plane graph}. If $\Sigma$ is a sphere, then we omit mentioning $\Sigma$, and say that $G$ is a \textit{$k$-planar} or a \textit{$k$-plane} graph. In this paper, we deal extensively with $(\Sigma,1)$-plane graphs. For any $(\Sigma,1)$-plane graph, we can always assume that no pair of edges incident to the same vertex cross each other, since the edges can be re-drawn without affecting $(\Sigma,1)$-planarity. Hence, every crossing in a $(\Sigma,1)$-plane graph has four distinct endpoints. In a $(\Sigma,1)$-plane graph, if two edges $e_1$ and $e_2$ cross each other, we use the notation $\{e_1,e_2\}$ to denote the crossing.  A simple $(\Sigma,1)$-planar graph is said to be \textit{maximal} if no edge can be added to the graph while maintaining $(\Sigma,1)$-planarity and simplicity. Consider a crossing $\{(u,v), (w,x)\}$ in a $(\Sigma,1)$-planar drawing of a graph $G$, with $c$ as the crossing point. If there is an uncrossed edge connecting two consecutive endpoints of the crossing, say $(u,w)$, and $\{(u,w), (w,c), (c,u)\}$ bounds a face of the drawing, then we say that $(u,w)$ is a \textit{kite-edge} of the crossing. In a maximal $(\Sigma,1)$-plane graph, the endpoints of every crossing induce the complete graph $K_4$. More generally, a graph is called \textit{full $(\Sigma,k)$-planar} if it has a $(\Sigma,k)$-planar drawing such that every crossing is full. If one allows for parallel edges, then these graphs can be drawn such that there is a kite-edge connecting every pair of consecutive endpoints of a crossing.   

\subparagraph{Cops and Robbers.} \textit{Cops and Robbers} is a game played on a graph $G$ where a set of cops try to capture a robber. The game proceeds in \textit{rounds}, where each round consists first of the \textit{cops' turn} and then the \textit{robber's turn}. We assume that the game is played with perfect information---all players know all other players' position at every time instant. \edit{In the first round, the cops initially choose to place themselves on a subset of vertices, and thereafter, the robber chooses a vertex for himself. In subsequent rounds, in the cops' turn, every cop can choose to stay on the same vertex or move to an adjacent vertex. Similarly, in the robber's turn, the robber may choose to stay on the same vertex or move to an adjacent vertex.} In the classical version of the game, the robber is considered to be \textit{captured} if, at some point, the robber is on the same vertex as a cop. The \textit{cop-number} of a graph $G$, denoted by $c(G)$, is the minimum number of cops required to capture a robber. A family of graphs $\mathcal{G}$ is said to be \textit{cop-bounded} if there is an integer $p$ such that $c(G) \leq p$ for all graphs $G \in \mathcal{G}$; otherwise, $\mathcal{G}$ is said to be \textit{cop-unbounded}. A variant of this game which we study extensively is \textit{Distance $d$ Cops and Robbers}, where a robber is considered to be \textit{captured} if, at some point, the robber is within distance $d$ of some cop \cite{distance_cops_robbers}. When $d = 0$, this reduces to the classical Cops and Robbers game. The \textit{distance $d$ cop-number} of $G$, denoted by $c_d(G)$, is the minimum number of cops required to capture the robber within distance $d$.    
\section{Distance \texorpdfstring{$d$}{d} Cops and Robbers.}\label{sec: distance d cops and robber}

At the outset, we state one of the simplest yet frequently used observation in our study of distance $d$ cop-numbers.  

\begin{observation}\label{obs: distance cop number monotonic}
    For any graph $G$ and integers $d_1 \geq d_2 \geq 0$, we have $c_{d_1}(G) \leq c_{d_2}(G)$.
\end{observation}

By \cref{obs: distance cop number monotonic}, $c_d(G) \leq c(G)$ for all values of $d \geq 0$. However, the two types of cop numbers are connected in a more interesting way. For any integer $k > 0$, let $G^{(k)}$ denote the graph obtained by replacing each edge of $G$ with a path on $k$ edges. Equivalently, the graph $G^{(k)}$ is obtained from $G$ by subdividing each edge $(k-1)$ times. 

\begin{theorem}[Lemma 9 in \cite{distance_cops_robbers}]\label{thm: distance cop-number 2d+1}
For any graph $G$ and any integer $d \geq 0$, $c(G) \leq c_d(G^{(2d+1)}) \leq c(G)+1$.
\end{theorem}

It is worth mentioning that there is no direct relation between $c(G)$ and $c_d(G)$. In fact, for any pair of integers $d,m \geq 1$, there exists a graph $G$ such that $c_d(G) = 1$ but $c(G) \geq m$ \cite{distance_cops_robbers}. Despite this, some important results from the classical Cops and Robbers game can be extrapolated to the distance version of the game. For instance, one of the well-known results concerns the cop-number of a graph $G$ and its retracts. A subgraph $H \subseteq G$ is a \textit{retract} of $G$ if there is a function $f: V(G) \mapsto V(H)$ such that $f(v) = v$ for every $v \in V(H)$ and if $(u,v) \in E(G)$, then $(f(u), f(v)) \in E(H)$. It is well-known that if $H$ is a retract of $G$, then $c(H) \leq c(G) \leq \max\{c(H), c(G \setminus H) + 1\}$ \cite{berarducci1993cop}. One can easily extend this to distance $d$ Cops and Robbers, as shown in \cref{theorem: retract at a distance}. 

\begin{theorem}[Theorems 3.1 and 3.2 in \cite{berarducci1993cop}] \label{theorem: retract at a distance}
    If $H$ is a retract of a graph $G$, then $c_d(H) \leq c_d(G) \leq \max\{c_d(H), c_d(G\setminus H) + 1\}$.
\end{theorem}

The proof for \cref{theorem: retract at a distance} is similar to the proof for the special case $d = 0$. For values of $d > 0$, we make use of the fact that if there is a path $P$ from $u$ to $v$ in $G$, then $H$ contains a path $P'$ from $f(u)$ to $f(v)$ of length at most $|P|$. 

We now state another result that is a distance analogue of the following classical result: $c(G) \leq c(G^{(k)}) \leq c(G) + 1$ for any $k \geq 1$ \cite{joret}. Considering that the reader may not find it straight-forward to extend the proof in \cite{joret} to \cref{theorem: subdivision}, and given the importance of this theorem for this paper, we believe that it will be useful to provide a complete proof of \cref{theorem: subdivision}.
  
\begin{theorem}\label{theorem: subdivision}
    For any graph $G$ and integers $d\geq 0$, $k \geq 1$, we have $c_{\ceil{d/k}}(G) \leq c_d(G^{(k)}) \leq c_{\floor{d/k}}(G) + 1$.
\end{theorem}

\begin{proof}
We will first show that $c_{\ceil{d/k}}(G) \leq c_d(G^{(k)})$. We simulate each robber-move in $G$ as a sequence of $k$ robber-moves in $G^{(k)}$, and the concomitant sequence of $k$ cop-moves in $G^{(k)}$ as a single cop-move in $G$. Therefore, a single round of $G$ is simulated by a block of $k$ rounds in $G^{(k)}$. Let $\mathcal{U} = \{U_1, \dots, U_p\}$ and $\mathcal{U'} = \{U_1', \dots, U_{p}'\}$ be two sets of $p$ cops playing in $G$ and $G^{(k)}$, respectively, where $p = c_d(G^{(k)})$. In the first round, the $p$ cops of $\mathcal{U}$ are positioned such that each cop $U_i$ is placed on a vertex of $G$ closest to the vertex occupied by $U_i'$ in $G^{(k)}$. The robber is placed in $G^{(k)}$ on the same vertex as the robber occupies in $G$. In any subsequent round, suppose that the robber moves along an edge $(a,b)$ in $G$. Then, we simulate $k$ rounds in $G^{(k)}$ where the robber moves along the path from $a$ to $b$ consisting of $k$ edges. Corresponding to the moves of any cop $U_i'$ in $G^{(k)}$ in these $k$ rounds, the cop $U_i$ can make a single move such that $U_i$ occupies a vertex of $G$ closest to the vertex occupied by $U_i'$. Suppose that after some blocks of $k$ rounds, the robber is captured in $G^{(k)}$ by a cop $U_i'$. Then, after the cops in $G$ have made their moves, the cop $U_i$ must be within distance $\ceil{d/k}$ from the robber in $G$. (If the robber in $G^{(k)}$ is captured in the midst of $k$ rounds while going from a vertex $a \in V(G)$ to a vertex $b \in V(G)$, where $(a,b) \in E(G)$, then the robber is made to complete the $k$ rounds and reach $b$, while the cop that captured the robber is made to follow a shortest path to the robber until the $k$ rounds are complete.)

Next, we show that $c_d(G^{(k)}) \leq c_{\floor{d/k}}(G) + 1$. Here, we simulate the robber moves of $G^{(k)}$ in $G$, and the cop moves of $G$ in $G^{(k)}$. The main strategy is to simulate $k$ rounds of the game in $G^{(k)}$ as a single round in $G$. Let $\mathcal{U} = \{U_1, \dots, U_p\} \cup \{U^*\}$ and $\mathcal{U'} = \{U_1', \dots, U_{p}'\}$ be two sets of cops playing in $G^{(k)}$ and $G$, respectively, where $p = c_{\floor{d/k}}(G)$. Initially, the $p$ cops $U_i$ of $G^{(k)}$ are placed on the same vertices as the cops $U_i'$ in $G^{(k)}$, whereas the cop $U^*$ is placed on an arbitrary vertex $s$. If $t$ is the initial vertex of the robber in $G^{(k)}$, then a robber is placed in the graph $G$ on a vertex of $G$ closest to $t$. After the first round, the role of the cop $U^*$ will be to track down the robber first along a shortest $(s,t)$-path in $G$, and then to trace the moves taken by the robber after the first round. This ensures that the robber does not stay at the same vertex or back-track its moves, except for a finite number of rounds. Henceforth, we may assume that if the robber in $G^{(k)}$ is presently on a vertex $v$ of $G$, then after $k$ moves, the robber is on another vertex $w$ of $G$ where $(v,w) \in E(G)$. And so, one can simulate this as a single robber-move in $G$ from $v$ to $w$. On the other hand, if a cop $U_i'$ in $G$ moves from vertex $a$ to $b$, then $U_i$ takes the path on $k$ edges from $a$ to $b$ in $G^{(k)}$. Therefore, every block of $k$ rounds in $G^{(k)}$ can be simulated as a single round in $G$. After some rounds in $G$, if a cop $U_i$ captures the robber within distance $\floor{d/k}$, then the distance between $U_i'$ and the robber, after the corresponding block of $k$ rounds is completed in $G^{(k)}$, is at most $k\floor{d/k} \leq d$.
\end{proof}

\subsection{\texorpdfstring{$(\Sigma,1)$}{Sigma1}-Planarisation and Kite-Augmentation.}

We use \cref{theorem: retract at a distance,theorem: subdivision} to define two simple operations for graphs drawn on surfaces, both of which are foundational to this paper. The two operations, called $(\Sigma,1)$-planarisation and kite-augmentation, are described below. 

\begin{definition}[$(\Sigma,1)$-Planarisation]  
For any graph $G$ drawn on a surface $\Sigma$, let the \emph{$(\Sigma,1)$-planarisation} of $G$ be a graph $G^{(k)}$ such that $G^{(k)}$ is $(\Sigma,1)$-plane.
\end{definition}

Every graph $G$ drawn on a surface $\Sigma$ has a $(\Sigma,1)$-planarisation: set $k$ equal to the minimum integer such that every edge is crossed at most $k$ times, and place $(k-1)$ subdivision vertices on each edge around crossing points so that $G^{(k)}$ is $(\Sigma,1)$-planar. Note that $(\Sigma,1)$-planarisation of a graph is not unique, since subdivisions of $(\Sigma,1)$-plane graphs are also $(\Sigma,1)$-plane. When $\Sigma$ is a sphere, then the resulting graph is 1-plane; in this case we simply write \textit{1-planarisation} instead of $(\Sigma,1)$-planarisation. In \cite{optimal1plane}, the process of 1-planarisation was used to prove that 1-planar graphs are cop-unbounded. Briefly, notice that by \cref{theorem: subdivision}, if $G^{(k)}$ is a 1-planarisation of $G$, then $c(G) \leq c(G^{(k)}) \leq c(G) + 1$. Since every graph $G$ has a 1-planarisation $H$, and the universal set of graphs is cop-unbounded, so is the set of all 1-planar graphs. 

Recall from \cref{sec: prelims} that a kite edge at a crossing of a $(\Sigma,1)$-plane graph is an uncrossed edge connecting two consecutive endpoints of a crossing such that the edge bounds a face incident to the crossing point. In this paper, we shall see that the presence of kite edges, or paths of uncrossed edges connecting consecutive endpoints of crossings, are especially useful to bound cop-numbers of 1-planar graphs. To this end, we introduce a class of $(\Sigma,1)$-planar drawings on surfaces called kite-augmented $(\Sigma,1)$-planar drawings.

\begin{definition}[Kite-augmented $(\Sigma,1)$-plane graph.]
    A $(\Sigma,1)$-plane graph is said to be \emph{kite-augmented} if, for every pair of consecutive endpoints $u,v$ of a crossing at a point $c$, there is a shortest $(u,v)$-path $\kappa_{uv}$ \edit{such that each edge of $\kappa_{uv}$ is uncrossed, each internal vertex of $\kappa_{uv}$ has degree 2, and $\{(c,u)\} \cup \kappa_{uv} \cup \{(v,c)\}$ bounds a face of the drawing}.
\end{definition}

The above definition lends itself to a procedure called kite-augmentation, defined below. 

\begin{definition}[Kite-augmentation]\label{def: kite-augmentation}
For any $(\Sigma,1)$-plane graph, let \emph{kite-augmentation} refer to the following process: for every pair of consecutive endpoints $u$ and $v$ of every crossing in the drawing, add a kite edge between $u$ and $v$ (if it does not already exist), and subdivide it to get a path $\kappa_{uv}$ of length equal to that of a shortest $(u,v)$-path in the original drawing.
\end{definition}

\begin{figure}
     \centering
     \begin{subfigure}[b]{0.45\textwidth}
         \centering
         \includegraphics[scale = 0.6, page = 1]{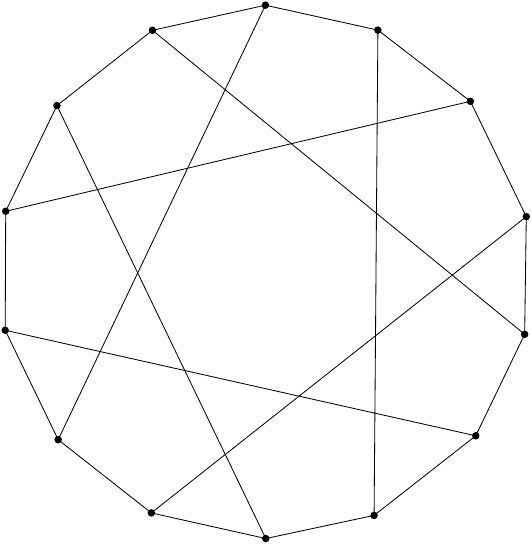}
     \end{subfigure}
     \hfill
     \begin{subfigure}[b]{0.45\textwidth}
         \centering
         \includegraphics[scale = 0.6, page = 2]{heawood.pdf}
     \end{subfigure}
    \caption{\edit{On the left is a drawing of a graph $G$, and on the right is a drawing of $G^\boxtimes$. The numbers alongside the edges indicate the lengths of the paths $\kappa_{uv}$.}}
    \label{fig: heawood}
\end{figure}

For any graph $G$ drawn on a surface $\Sigma$, we use the notation $G^\boxtimes$ to denote a graph obtained by $(\Sigma,1)$-planarising $G$ and then kite-augmenting the resulting graph \edit{(see \cref{fig: heawood} for an example)}. If $G^{(k)}$ is a $(\Sigma,1)$-plane graph, and $G^\boxtimes$ is obtained by kite-augmenting $G^{(k)}$, then $G^{(k)}$ is a retract of $G^\boxtimes$ since there is a homomorphism from each path $\kappa_{uv}$ to a shortest $(u,v)$-path in $G^{(k)}$. This observation, together with \cref{theorem: retract at a distance,theorem: subdivision} gives us the following corollary.

\begin{corollary}[Corollary of \cref{theorem: retract at a distance,theorem: subdivision}]\label{corollary: 1-planarisation cop-number}
Let $G$ be any graph drawn on a surface $\Sigma$, $G^{(k)}$ be a $(\Sigma,1)$-planarisation of $G$, and $G^\boxtimes$ be the kite-augmentation of $G^{(k)}$. Then $c_{\ceil{d/k}}(G) \leq c_d(G^\boxtimes) \leq c_{\floor{d/k}}(G) + 1$ for any integer $d \geq 0$.  
\end{corollary}

\begin{proof}
   From \cref{theorem: subdivision}, we have $c_{\ceil{d/k}}(G) \leq c_d(G^{(k)}) \leq c_{\floor{d/k}}(G) + 1$. Since $G^{(k)}$ is a retract of $G^\boxtimes$, we have from \cref{theorem: retract at a distance} that $c_d(G^{(k)}) \leq c_d(G^\boxtimes) \leq \max\{c_d(G^{(k)}), c_d(G^\boxtimes \setminus G^{(k)}) + 1\}$. As $G^\boxtimes$ is obtained by kite-augmenting $G^{(k)}$, the graph $G^\boxtimes \setminus G^{(k)}$ is a collection of vertex-disjoint paths; hence $c_d(G^\boxtimes \setminus G^{(k)}) = 1$. This implies that $c_d(G^\boxtimes) \leq \max\{c_d(G^{(k)}), 2\} \leq \max\{c_{\floor{d/k}}(G) + 1, 2\} = c_{\floor{d/k}}(G) + 1$. By combining the above inequalities, we get $c_{\ceil{d/k}}(G) \leq c_d(G^\boxtimes) \leq c_{\floor{d/k}}(G) + 1$. 
\end{proof}

The process of $(\Sigma,1)$-planarisation combined with kite-augmentation leads to several interesting consequences. One such consequence is presented in \cref{thm: single crossing per face}. It has long been known that planar graphs have cop-number at most 3 \cite{AignerFromme}. An interesting question is whether the set of all graphs $G$ that can be drawn with at most one crossing in each face of its skeleton $\mathrm{sk}(G)$ (i.e., the subgraph induced by the set of all uncrossed edges) has a small cop-number. Using the process of 1-planarisation and kite-augmentation, we can show that this is false. Indeed, one can even assume that no two crossings are close to each other in terms of face-distances. For a precise formulation, we need some definitions first. For any graph $G$ drawn on the sphere, consider the graph obtained by first planarising $G$ (i.e., adding dummy vertices at crossing points), and then inserting a \textit{face-vertex} inside each face and making it adjacent to all vertices on the face-boundary. Define the \textit{radial graph} of $G$ as the plane bipartite subgraph $R(G)$ induced by all edges incident with face-vertices. Let the \textit{face-distance} between any two crossings be the number of face-vertices on a shortest path in $R(G)$ connecting the two crossing points. 

\begin{theorem}\label{thm: single crossing per face}
For any three integers $d,f,m \geq 0$, there exists a 1-plane graph $G$ such that each face of $\mathrm{sk}(G)$ contains at most one crossing, the face-distance between any two crossing points is at least $f$, and $c_d(G) \geq m$. 
\end{theorem}

\begin{proof}
Consider any graph $H$ with $c(H) \geq m$. From \cref{thm: distance cop-number 2d+1}, we know that $c_d(H^{(2d+1)}) \geq m$. Let $G'$ be a graph obtained by 1-planarising $H^{(2d+1)}$ and then kite-augmenting it. Therefore, $G'$ is a 1-planar graph such that $\mathrm{sk}(G')$ has at most one crossing per face. Using \cref{corollary: 1-planarisation cop-number} and \cref{obs: distance cop number monotonic} ($d_1 \leq d_2$ implies $c_{d_1}(G) \geq c_{d_2}(G)$), we can see that $m \leq c_d(H^{(2d+1)}) \leq c_d(G')$. To make the pairwise face-distance between the crossing points at least $f$, we repeat the following process $\ceil{f/2}$ times: Subdivide each edge twice to get a 1-planar drawing such that the endpoints of every crossing consist only of the newly added subdivision vertices, and then kite-augment the graph. This increases the face-distance between any two crossing points by 2, and after repeating $\ceil{f/2}$ times, the face-distance becomes at least $f$. Let $G$ be the resulting graph. By repeatedly applying \cref{corollary: 1-planarisation cop-number} (and \cref{obs: distance cop number monotonic}), we can verify that $c_d(G) \geq m$.  
\end{proof}

\subsection{Parameterizing \texorpdfstring{$d$}{d} through Distances Between Crossing Edges.}

Our study of distance $d$ Cops and Robbers for graphs drawn on surfaces focuses primarily on those values of $d$ that in some way capture how close crossing pairs of edges are in the drawing. To this end, we define two parameters $X(G)$ and $x(G)$:

\begin{definition}[$x(G)$ and $X(G)$.]\label{def: x and X}
For any fixed drawing of a graph $G$ on a surface, $x(G)$ is the smallest integer such that for every crossing in the drawing, there is a path of length at most $x(G)$ connecting \emph{some} pair of consecutive endpoints of the crossing. On the other hand, $X(G)$ is the smallest integer such that for every crossing in the drawing, there is a path of length at most $X(G)$ connecting \emph{every} pair of consecutive endpoints of the crossing.
\end{definition}

These two parameters are closely related to a structure called ribbon associated with crossings (first defined in \cite{biedl2024parameterized}): For any crossing pair of edges $e_1$ and $e_2$, a \textit{ribbon} at the crossing is a path with $e_1$ as the first edge and $e_2$ as the last edge. Every crossing in a graph has a ribbon of length at most $x(G) + 2$. Since any pair of consecutive endpoints of a crossing can be connected by a sub-path of a ribbon at the crossing, $X(G) \leq x(G) + 2$. Therefore, $X(G) \in \{x(G), x(G)+1, x(G)+2\}$. Our intention in differentiating between $x(G)$ and $X(G)$ is to capture the differences arising in cop-numbers due to variations in crossing configurations. For instance, in a maximal 1-planar graph $G$ where the endpoints of every crossing induce $K_4$, we have $x(G) = X(G) = 1$, and $c(G) = 3$ \cite{optimal1plane}. However, in a 1-planar graph without $\times$-crossings where the endpoints induce at least 3 edges, we have $x(G) = 1$ (and $X(G) \leq 3$), and $c(G) \leq 21$ \cite{tcs_1planar}. 

In \cref{lemma: bounds on x and X}, we give useful bounds on $X(G^{(k)})$ and $x(G^{(k)})$ in terms of the corresponding parameters for $G$. These will be used in \cref{thm: 1-planarisation}, and in later sections of the paper.  

\begin{lemma}\label{lemma: bounds on x and X}
For any graph $G$ and integer $k \geq 1$, we have $ k\cdot X(G) \leq X(G^{(k)}) \leq k\cdot X(G) + 2\floor{k/2}$  and $ k \cdot x(G) \leq x(G^{(k)}) \leq k \cdot x(G) + 2(k-1)$.
% \begin{enumerate}
%     \item $ k\cdot X(G) \leq X(G^{(k)}) \leq k\cdot X(G) + 2\floor{k/2}$ 
%     \item $ k \cdot x(G) \leq x(G^{(k)}) \leq k \cdot x(G) + 2(k-1)$ 
% \end{enumerate}  
\end{lemma}

\begin{proof}
Let us call an edge $(u,v) \in E(G^{(k)})$ an \textit{internal edge} of $(u',v') \in E(G)$ if $P(u',v') = P(u'\dots u) \cup (u,v) \cup P(v\dots v')$, where $P(u',v')$ is the path on $k$ edges in $G^{(k)}$ that replaced the edge $(u',v') \in E(G)$. For every crossing $\{(u,v), (w,x)\}$ in $G^{(k)}$, there is a crossing $\{(u', v'), (w', x')\}$ of $G$ such that $(u,v)$ and $(w,x)$ are internal edges of $(u',v')$ and $(w',x')$ respectively. It is easy to see that $k \cdot X(G) \leq X(G^{(k)})$ and $k \cdot x(G) \leq x(G^{(k)})$, since every crossing in $G^{(k)}$ consists of edges that are internal to another crossing in $G$, and any path $P$ in $G$ has a corresponding path of length $k|P|$ in $G^{(k)}$. It is slightly less trivial to show the upper bounds on $X(G^{(k)})$ and $x(G^{(k)})$. For the rest of this proof, fix an arbitrary crossing $\{(u, v), (w, x)\}$ of $G^{(k)}$, and let these edges be internal to a corresponding crossing $\{(u',v'), (w',x')\}$ of $G$. 

To show the upper bound on $X(G^{(k)})$, we will show that for any pair of consecutive endpoints, say $u$ and $w$, there exists a $(u,w)$-path in $G^{(k)}$ of length at most $k \cdot X(G) + 2\floor{k/2}$. Notice that either the $(u,u')$-path or the $(u,v')$-path in $G^{(k)}$ has length at most $\floor{k/2}$. Likewise, either the $(w,w')$-path or the $(w,x')$-path has length at most $\floor{k/2}$. These two paths can be combined with a path of length at most $k \cdot X(G)$, which exists between every pair of consecutive endpoints of the crossing $\{(u', v'), (w', x')\}$. Hence, there is a $(u,w)$-path of length at most $k\cdot X(G) + 2\floor{k/2}$.

We now show the upper bound on $x(G^{(k)})$. Assume without loss of generality that there is a path of length at most $x(G)$ between $u'$ and $w'$ in $G$. The lengths of the $(u,u')$-path and the $(w,w')$-path are each at most $k-1$. When these paths are combined with the $(u',w')$-path in $G^{(k)}$ of length $k \cdot x(G)$, we get a $(u,w)$-path of length at most $k \cdot x(G) + 2(k-1)$.
\end{proof}

In \cref{thm: 1-planarisation}, we give a relation between distance $d$ cop-numbers of $G$ and $G^\boxtimes$, where $d$ is measured through the parameters $x(G)$ and $X(G)$.

\begin{theorem}\label{thm: 1-planarisation}
For all graphs $G$ drawn on a surface $\Sigma$, all graphs $G^\boxtimes$ obtained by $(\Sigma,1)$-planarising and kite-augmenting $G$, and all real numbers $\alpha > 0$:
\begin{enumerate}[(a)]
    \item $c_{\ceil{\alpha (x(G)+2)}}(G) \leq c_{\ceil{\alpha x(G^\boxtimes)}-1}(G^\boxtimes) \leq c_{\floor{\alpha x(G)}-1}(G) + 1$
    \item $c_{\ceil{\alpha (X(G)+1)}}(G) \leq c_{\ceil{\alpha X(G^\boxtimes)}-1}(G^\boxtimes) \leq c_{\floor{\alpha X(G)}-1}(G) + 1$
\end{enumerate} 
\end{theorem}

\begin{proof}
Fix a graph $G^\boxtimes$, and consider the graph $(\Sigma,1)$-plane graph $G^{(k)}$ such that $G^\boxtimes$ is the kite-augmentation of $G^{(k)}$. For every pair of vertices $u,v \in G^{(k)}$, there is a $(u,v)$-path in $G^\boxtimes$ of length at most $p$ if and only if there is a $(u,v)$-path in $G^{(k)}$ of length at most $p$. Hence, $X(G^\boxtimes) = X(G^{(k)})$ and $x(G^\boxtimes) = x(G^{(k)})$. By setting $d = \ceil{\alpha X(G^\boxtimes)} - 1$ or $d = \ceil{\alpha x(G^\boxtimes)} - 1$ in \cref{corollary: 1-planarisation cop-number}, and using the bounds on $X(G^\boxtimes)$ and $x(G^\boxtimes)$ provided by \cref{lemma: bounds on x and X}, we arrive at the following inequalities:
{\allowdisplaybreaks
\begin{align}
    &\left \lceil \frac{\ceil{\alpha x(G^\boxtimes)} - 1}{k} \right \rceil \leq \left \lceil \frac{\ceil{\alpha (k x(G) + 2k)}}{k} \right \rceil = \left \lceil \frac{\ceil{ k (\alpha (x(G) + 2))}}{k} \right \rceil \leq \ceil{\alpha (x(G) + 2)} \\
    &\left \lfloor \frac{\ceil{\alpha x(G^\boxtimes)} - 1}{k} \right \rfloor \geq \left \lfloor \frac{\alpha kx(G) - 1}{k} \right \rfloor = \left \lfloor \alpha x(G) - 1/k \right \rfloor \geq \floor{\alpha x(G)}-1 \\
    &\left \lceil \frac{\ceil{\alpha X(G^\boxtimes)} - 1}{k} \right \rceil \leq \left \lceil \frac{\ceil{\alpha (k X(G) + k)}}{k} \right \rceil = \left \lceil \frac{\ceil{k (\alpha (X(G) + 1))}}{k} \right \rceil \leq \ceil{\alpha (X(G) + 1)} \\
    &\left \lfloor \frac{\ceil{\alpha X(G^\boxtimes)} - 1}{k} \right \rfloor \geq \left \lfloor \frac{\alpha kX(G) - 1}{k} \right \rfloor = \left \lfloor \alpha X(G) - 1/k \right \rfloor \geq \floor{\alpha X(G)}-1
\end{align}}

By setting $d = \ceil{\alpha x(G^\boxtimes)} - 1$ in \cref{corollary: 1-planarisation cop-number}, and using Equations (1) and (2), we have $c_{\ceil{\alpha (x(G)+2)}}(G) \leq c_{\left \lceil \frac{\ceil{\alpha x(G^\boxtimes)} - 1}{k} \right \rceil}(G) \leq c_{\ceil{\alpha x(G^\boxtimes)} - 1}(G^\boxtimes) \leq c_{\left \lfloor \frac{\ceil{\alpha x(G^\boxtimes)} - 1}{k} \right \rfloor}(G) \leq c_{\floor{\alpha x(G)}-1}(G) + 1$. Similarly, by setting $d = \ceil{\alpha X(G^\boxtimes)} - 1$ in \cref{corollary: 1-planarisation cop-number}, and using Equations (3) and (4), we have $c_{\ceil{\alpha (X(G)+1)}}(G) \leq c_{\left \lceil \frac{\ceil{\alpha X(G^\boxtimes)} - 1}{k} \right \rceil}(G) \leq c_{\ceil{\alpha X(G^\boxtimes)} - 1}(G^\boxtimes) \leq c_{\left \lfloor \frac{\ceil{\alpha X(G^\boxtimes)} - 1}{k} \right \rfloor}(G) \leq c_{\floor{\alpha X(G)}-1}(G) + 1$.
\end{proof}

\section{1-Planar Graphs}\label{sec: 1-planar}

In view of \cref{thm: 1-planarisation}, it is sufficient to restrict attention to kite-augmented $(\Sigma,1)$-plane graphs. In this section, we focus on kite-augmented 1-planar graphs (drawn on the sphere), while \cref{sec: graphs on surfaces} extends the discussion to kite-augmented $(\Sigma,1)$-plane graphs on surfaces. The entirety of this section is devoted to proving \cref{thm: cop number 1-planar weighted kite edge} and understanding its implications. 

\begin{theorem}\label{thm: cop number 1-planar weighted kite edge}
Let $G$ be any kite-augmented 1-plane graph. For $d \in \{X(G), x(G)\}$ and $\alpha \geq 1$, we have $c_{\ceil{\alpha d} - 1}(G) \leq 3 \cdot(2 \beta + 1)$ where 
\begin{equation*}
\beta = \begin{cases}
             d + 1  & \text{if } d = x(G) \text{ and } \alpha = 1 \\
             \left \lceil \frac{1}{2(\alpha - 1)}\right \rceil + 1   & \text{if } d = x(G) \text{ and } \alpha > 1
       \end{cases} \qquad
\beta = \begin{cases}
             d - 1  & \text{if } d = X(G) \text{ and } \alpha = 1 \\
             \left \lceil \frac{1}{2(\alpha - 1)} \right \rceil  & \text{if } d = X(G) \text{ and } \alpha > 1
       \end{cases}
\end{equation*}
\end{theorem}

The proof of \cref{thm: cop number 1-planar weighted kite edge} is structured along the lines of the proof that the cop-number of a planar graph is at most 3 \cite{AignerFromme, bonato_book}. For planar graphs, the essential idea is that 3 cops progressively guard larger and larger subgraphs, while maintaining as an invariant that the frontier between the guarded and unguarded subgraph is always either a shortest path or a cycle composed of two shortest paths. Unlike planar graphs that require a single cop to guard a shortest path, we require $2\beta + 1$ cops since the cops must also ensure that the robber does not cross any edge of the shortest path; this coarsely explains why the cop-number increases by a factor of $2\beta+1$ in \cref{thm: cop number 1-planar weighted kite edge}. 

\subsection{Guarding Shortest Paths in \texorpdfstring{$(s,t)$}{st}-Subgraphs.}

The choice of which shortest path or cycle to guard is made by selecting two vertices, say $s$ and $t$, and then finding a shortest $(s,t)$-path through the unguarded subgraph. Since we require that such a shortest path include at least one vertex from the unguarded subgraph, it may be necessary to exclude the edge $(s, t)$ when both $s$ and $t$ are already guarded and adjacent in $G$. In light of this, we find it convenient to view an unguarded subgraph as an $(s,t)$-subgraph, defined as follows. \edit{(For \cref{def: st subgraph}, recall the notation $\kappa_{ab}$ from \cref{def: kite-augmentation} for the path resulting from repeated subdivision of kite-edges.)

\begin{definition}[$(s,t)$-subgraphs of $G$]\label{def: st subgraph}
  Let $G$ be a kite-augmented 1-planar graph. Let $s$ and $t$ be two distinct vertices of $G$. A connected subgraph $H$ of $G$, with $s,t \in V(H)$, is an \emph{$(s,t)$-subgraph} of $G$ if, for any pair of distinct vertices $a,b \in V(H)$, except possibly for the pair $\{a,b\} = \{s,t\}$, if $\kappa_{ab} \subseteq G$, then $\kappa_{ab} \subseteq H$.  
\end{definition}
}

In \cref{lem: shortest paths do not cross}, we show that there always exists a shortest $(s,t)$-path in an $(s,t)$-subgraph such that no two edges of the shortest path cross each other. For brevity, we say that a subgraph $H$ of $G$ is \textit{self-crossing} if there exist a pair of edges of $H$ that cross each other; otherwise, $H$ is \textit{non-self-crossing}.   

\begin{lemma}\label{lem: shortest paths do not cross}
In any $(s,t)$-subgraph $H$ of a kite-augmented 1-planar graph $G$, there exists an $(s,t)$-path $P$ that is a shortest path in $H$ and is non-self-crossing.  
\end{lemma}

\begin{proof}
Let $P$ be a shortest $(s,t)$-path of $H$ with the minimum number of self-crossings. If $P$ has no self-crossings, then we are already done, so assume otherwise, for the sake of contradiction. We will construct another shortest $(s,t)$-path $P'$ with fewer self-crossings than $P$. Enumerate the vertices of $P$ based on their distances from $s$, and let $\{(u,v), (w,x)\}$ be a crossing of $G$ such that $\{(u,v), (w,x)\} \subseteq E(P)$, and up to renaming, let the four endpoints of the crossing appear in $P$ in the order $u,v,w,x$. As $G$ is kite-augmented, there exists a path $\kappa_{ux}$ in $G$. If $\kappa_{ux} \subseteq H$, then one can substitute the subpath $P(u \dots x)$ with $\kappa_{ux}$ to get another shortest path; i.e., $P' := P \setminus E(P(u \dots x)) \cup E(\kappa_{ux})$ is a shortest $(s,t)$-path in $H$ with fewer self-crossings. Since this leads to a contradiction, we must assume that $\kappa_{ux}$ is not a subgraph of $H$. Since $H$ is an $(s,t)$-subgraph and $\kappa_{ux}$ is a path of uncrossed edges, this can happen only when $\{u,x\} = \{s,t\}$. Let $P'$ be the path $\kappa_{uw} \cup \{(w,x)\}$. Since $|\kappa_{uw}| \leq |P(u \dots w)|$, we have $|P'| \leq |P|$. This again leads to a contradiction as $P'$ has fewer self-crossings than $P$.   
\end{proof}

As mentioned before, our intention is to guard paths and cycles in $(s,t)$-subgraphs of $G$ such that the robber can neither land on a vertex nor cross any edge of the path or cycle. This idea is formalised with the following definition.

\begin{definition}[Crossing-guarded subgraph]\label{def: crossing-guard}
For any integer $d \geq 0$, a graph $H$ is \emph{crossing-guarded at distance $d$ by a set $\mathcal{U}$ of cops} if: 
\begin{enumerate}[(a)]
    % \item All vertices of $H$ are distance $d$ guarded by $\mathcal{U}$, and
    \item \edit{For any vertex $v \in V(H)$, if the robber lands on $v$, then he is captured by a cop of $\mathcal{U}$, and}
    \item For any crossing $\{(u,v), (w,x)\}$ of $G$ with $(u,v) \in E(H)$, if the robber crosses the edge $(u,v)$ by moving from $w$ to $x$, then he is captured by a cop of $\mathcal{U}$ within distance $d$.
\end{enumerate}
\end{definition}

In \cref{lem: shortest path guarding}, we show that any shortest path in an $(s,t)$-subgraph of $G$ is crossing guardable at distance $\ceil{\alpha d}-1$ by $2 \beta + 1$ cops for the values of $\alpha$, $d$ and $\beta$ as stated in \cref{thm: cop number 1-planar weighted kite edge}.

\begin{lemma}\label{lem: shortest path guarding}
\edit{Let $H$ be an $(s,t)$-subgraph of $G$ containing the robber, where the robber is restricted to moving only along the edges of $H$.} Let $P$ be a shortest $(s,t)$-path in $H$. For $d \in \{x(G), X(G)\}$ and any $\alpha \geq 1$, $P$ is crossing-guardable at distance $\ceil{\alpha d}-1$ by a set $\mathcal{U}$ of cops where $|\mathcal{U}| \leq 2\beta + 1$ and
\begin{equation*}
\beta = \begin{cases}
             d + 1  & \text{if } d = x(G) \text{ and } \alpha = 1 \\
             \left \lceil \frac{1}{2(\alpha - 1)}\right \rceil + 1   & \text{if } d = x(G) \text{ and } \alpha > 1
       \end{cases} \qquad
\beta = \begin{cases}
             d - 1  & \text{if } d = X(G) \text{ and } \alpha = 1 \\
             \left \lceil \frac{1}{2(\alpha - 1)} \right \rceil  & \text{if } d = X(G) \text{ and } \alpha > 1
       \end{cases}
\end{equation*} 
\end{lemma}

\begin{proof}
Since the robber is restricted to $H$, and $P$ is a shortest path in $H$, all vertices of $P$ are guardable by a single cop (after a finite number of initial rounds). \edit{In other words, the robber can be captured by the cop whenever he lands on any vertex of $P$.} Moreover, this cop can guard all vertices of $P$ by remaining within the path $P$. (This is a very well-known and widely used result in Cops and Robbers; see \cite{AignerFromme, bonato_book}.) Let us denote this cop of $\mathcal{U}$ by $\mathcal{U}^*$. To ensure that $P$ is crossing-guardable, we set up the remaining $2\beta$ cops of $\mathcal{U}$, within a finite number of rounds, so that the following configuration is maintained at all times thereafter: 

\medskip
\noindent \textbf{Case $\alpha = 1$:} There is a cop at every vertex within distance $\beta$ of $\mathcal{U}^*$.

\noindent \textbf{Case $\alpha > 1$:} There is a cop at every vertex at distance exactly $m \cdot 2\ceil{\alpha d - d}$ from $\mathcal{U}^*$, for every $1 \leq m \leq \beta$.

\medskip
After the cops of $\mathcal{U}$ have been configured this way, at any given point in time, let the \textit{range of $\mathcal{U}$} be the minimal-length subpath of $P$ that contains all the cops of $\mathcal{U}$. 

\begin{observation}\label{obs: range of P}
At any point in time, if the robber is within distance $d$ of some vertex in the range of $\mathcal{U}$, then there is some cop of $\mathcal{U}$ within distance $\ceil{\alpha d}$ of the robber. 
\end{observation}

\begin{proof}
The observation is trivial when $\alpha = 1$ since all vertices in the range of $\mathcal{U}$ are occupied by a cop. When $\alpha > 1$, there is at least one cop of $\mathcal{U}$ in every subpath of length $2\ceil{\alpha d - d}$ in the range of $\mathcal{U}$. Hence, there exists a cop of $\mathcal{U}$ within distance $\ceil{\alpha d - d} + d = \ceil{\alpha d}$ from the robber. 
\end{proof}

\begin{observation}\label{obs: range and d}
    When $d = x(G)$ and $d = X(G)$, all vertices within distance $d+1$ and $d-1$, respectively, of $\mathcal{U}^*$ are in the range of $\mathcal{U}$.
\end{observation}

\begin{proof}
The observation is trivial when $\alpha = 1$, given that $\beta$ is precisely equal to $d+1$ or $d-1$ when $d = x(G)$ and $d = X(G)$ respectively. When $\alpha > 1$, there are $\beta \geq \left \lceil \frac{1}{2(\alpha - 1)}\right \rceil$ cops placed at regular intervals of length $2\ceil{\alpha d -d}$. Since $\left \lceil \frac{1}{2(\alpha - 1)}\right \rceil \cdot 2\ceil{\alpha d -d} \geq d$, the observation holds.  
\end{proof}

To complete the proof of \cref{lem: shortest path guarding}, we analyze the cases $d = X(G)$ and $d = x(G)$ separately. All vertices of $P$ are already guarded by $\mathcal{U}^*$, so we only consider edges of $P$ that are involved in crossings. Consider any crossing $\{(u,v), (w,x)\}$, where  $(u,v) \in E(P)$ and $w \in V(H)$, and assume that the robber has played his turn and landed on $w$. In both cases, we will show that the robber cannot move to $x$ without being captured by a cop of $\mathcal{U}$.

\medskip
\textbf{Case $d = x(G)$:} We first show that $|\kappa_{wu}| > d$ and $|\kappa_{wv}| > d$. Assume for contradiction that $|\kappa_{wu}| \leq d$ (the argument is symmetric for $|\kappa_{wv}| \leq d$). Consider the position of $\mathcal{U}^*$ just before the robber moved to $w$. At this point of time, the robber is within distance $d+1$ from $u$. Since $\mathcal{U}^*$ guards all vertices of $P$, it must also be within distance $d+1$ of $u$. By \cref{obs: range and d}, this implies that $u$ was in the range of $\mathcal{U}$ just before the robber moved to $w$. By \cref{obs: range of P}, there is a cop of $\mathcal{U}$ that is within distance $\ceil{\alpha d}$ from $w$. Therefore, after the robber lands on $w$, then this cop can capture the robber from distance $\ceil{\alpha d} - 1$.  

Since $|\kappa_{wu}| > d$ and $|\kappa_{wv}| > d$, at least one of $\kappa_{ux}$ and $\kappa_{vx}$ has length at most $d$; up to symmetry, assume that $|\kappa_{ux}| \leq d$. This implies that $|\kappa_{wu}| \leq d+1$. So, after the robber lands on $w$, the cop $\mathcal{U}^*$ moves so that it comes within distance $d+1$ of $u$, by which $u$ will now belong to the range of $\mathcal{U}$ (\cref{obs: range and d}). By \cref{obs: range of P}, if the robber moves to $x$, the distance to $u$ becomes at most $d$, and the cop closest to the robber can capture him within distance $\ceil{\alpha d}-1$. 

\medskip
\textbf{Case $d = X(G)$:} After the robber lands on $w$, then by virtue of $|\kappa_{wu}| \leq d$ and $|\kappa_{wv}| \leq d$, the cop $\mathcal{U}^*$ moves such that at least one of $u$ or $v$ is within distance $d-1$ of $\mathcal{U}^*$. By \cref{obs: range and d}, at least one of $u$ and $v$ in the range of $\mathcal{U}$. In the next step, if the robber crosses the edge $(u,v)$ and moves to $x$, then the cop closest to the robber can capture him within distance $\ceil{\alpha d}-1$ (\cref{obs: range of P}).
\end{proof}

\subsection{Cop-Strategy for Kite-Augmented 1-Planar Graphs.}

We now use \cref{lem: shortest path guarding} on guarding shortest paths to prove \cref{thm: cop number 1-planar weighted kite edge}. For this, we use ideas from \cite{bonato_book} for the proof that planar graphs have cop-number at most 3, and from \cite{tcs_1planar} for the modifications needed to handle crossings in 1-planar graphs. We first give an overview of the proof, and then provide a formal description in detail in \cref{lemma: 3 times U_d}. 

We consider three sets of $2\beta + 1$ cops, and proceed in iterations, where every iteration begins with a path or a cycle being crossing-guarded at distance $\ceil{\alpha d}-1$. (For the rest of this discussion, we use the term ``guard'' as a shortcut to saying ``crossing-guard at distance $\ceil{\alpha d}-1$''.) We ensure that the paths and cycles are non-self-crossing, so that they trace simple paths and simple cycles on the sphere. These paths and cycles serve as the frontiers between the guarded subgraph and the unguarded subgraph of $G$, and these are the only subgraphs that are actively guarded at any point in time. A path requires only one set of cops to guard, and a cycle requires two sets of cops; hence at least one of three sets of cops is free at the beginning of every iteration.

\begin{figure}
     \centering
     \begin{subfigure}[b]{0.45\textwidth}
         \centering
         \includegraphics[scale = 0.75, page = 1]{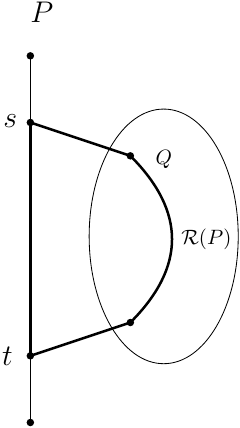}
     \end{subfigure}
     \hfill
     \begin{subfigure}[b]{0.45\textwidth}
         \centering
         \includegraphics[scale = 0.75, page = 2]{Guarding_configurations.pdf}
     \end{subfigure}
    \caption{The robber territory, denoted by $\mathcal{R}(\cdot)$, is progressively cut down across iterations, where each iteration is marked by either a path $P$ or a cycle $C$ being guarded. For each successive iteration, the new path or cycle to be guarded is determined by choosing a shortest $(s,t)$-path $Q$ to be guarded in the $(s,t)$-subgraph (the region enclosed within bold edges).}
    \label{fig: guarding configurations}
\end{figure}

Every iteration begins by choosing a shortest path to be guarded by the free set of cops. For this, two special vertices $s$ and $t$ are identified that have a neighbour in the robber territory---the maximal subgraph of $G$ within which the robber can move without getting captured. Then, we consider the $(s,t)$-subgraph formed by the union of the robber territory together with the set of all edges connecting $\{s,t\}$ and the robber territory. Now, we guard a shortest non-self-crossing $(s,t)$-path in the $(s,t)$-subgraph identified above. This path may be the only path actively guarded, or may be combined with a previously guarded path to form a new cycle to be actively guarded (\cref{fig: guarding configurations}). Once the path has been guarded, the current iteration ends, at least one set of cops is freed, and the next iteration begins. At the end of every iteration, we shall maintain the invariant that the shortest $(s,t)$-path chosen to be guarded contains at least one vertex of the robber territory. Therefore, every iteration ends with at least one more vertex of the graph being guarded, and after a finite number of iterations, the robber is captured. 

\medskip
We now give a detailed and formal treatment of the above overview in \cref{lemma: 3 times U_d}.

\begin{lemma}\label{lemma: 3 times U_d}
Let $G$ be a kite-augmented 1-plane graph. For $d \in \{x(G), X(G)\}$, $\alpha \geq 1$, and any $(s,t)$-subgraph $H$ of $G$, let $|\mathcal{U}|$ be the smallest integer such that any shortest $(s,t)$-path $P \subseteq H$ is crossing-guardable at distance $\ceil{\alpha d}-1$ by a set $\mathcal{U}$ of cops. Then $c_{\ceil{\alpha d}-1}(G) \leq 3|\mathcal{U}|$.
\end{lemma}

\begin{proof}
Our main strategy is to proceed in iterations, where every iteration begins with a new path or cycle being crossing-guarded at distance $\ceil{\alpha d}-1$ by one or two sets of $|\mathcal{U}|$ cops, respectively. For the rest of this discussion, we use the term ``guard'' as a shortcut to saying ``crossing-guard at distance $\ceil{\alpha d}-1$''. For every subgraph (path or cycle) of $G$ that is guarded, there is a corresponding robber territory, which is defined below. 

\begin{definition}[Robber territory]\label{def: robber territory}
    Let $H$ be a subgraph of $G$ that is guarded. Then the \emph{robber territory of $H$}, denoted $\mathcal{R}(H)$, is the maximal subgraph of $G$ containing the robber such that the robber can visit any vertex of $\mathcal{R}(H)$ using any edge of $\mathcal{R}(H)$ without getting captured. 
\end{definition}

As $H$ is crossing-guarded, no edge of $\mathcal{R}(H)$ crosses an edge of $H$. Given that we will repeatedly use this fact, we make it explicit through \cref{obs: robber territory}.

\begin{observation}\label{obs: robber territory}
For any guarded subgraph $H$, no edge of $\mathcal{R}(H)$ crosses an edge of $H$.
\end{observation}

The cycles and paths that are chosen to be guarded in each iteration are constructed by first choosing two special vertices $s$ and $t$, with at least one of them belonging to the current path or cycle being guarded, and then finding a shortest $(s,t)$-path through the robber territory. In order to maintain the invariant that all paths and cycles are non-self-crossing, $s$ and $t$ must be chosen such that there is (loosely speaking) an uncrossed edge connecting them to some vertex of the robber territory. This idea is made more formal through \cref{def: non-crossing adjacency} (both the definition and notation are adapted from \cite{tcs_1planar}). For simplicity, we shall henceforth regard every path $\kappa_{uv}$, connecting two consecutive endpoints $u,v$ of a crossing, as an uncrossed edge $(u,v)$ with weight $|\kappa_{uv}|$. All other edges of $G$ have weight 1.  

\begin{definition}[Crossing and non-crossing adjacency to $\mathcal{R}(H)$]\label{def: non-crossing adjacency}
Let $H$ be a guarded subgraph and $v \in V(H)$ be a vertex adjacent to a vertex in $\mathcal{R}(H)$. Let $S_H(v)$ be the set of all edges incident to $v$ and some vertex of $\mathcal{R}(H)$ that do not cross an edge of $H$. If $S_H(v)$ is non-empty, then we say that $v$ is non-crossing adjacent to $\mathcal{R}(H)$; else we say that $v$ is crossing adjacent to $\mathcal{R}(H)$.
\end{definition}

In \cref{obs: adjacency}, we show that one can always find a vertex that is non-crossing adjacent to the robber territory. Indeed, if $v \in V(H)$ is crossing adjacent to $\mathcal{R}(H)$, then there is an edge of $E(H)$ such that both its endpoints are non-crossing adjacent to $\mathcal{R}(H)$.

\begin{observation}\label{obs: adjacency}
 Let $H$ be a guarded subgraph of $G$. If $v \in V(H)$ is crossing adjacent to a vertex $w \in \mathcal{R}(H)$ so that $\{(v,w), (x,y)\}$ is a crossing in $G$ with $(x,y) \in E(H)$, then $x$ and $y$ are non-crossing adjacent to $\mathcal{R}(H)$.
\end{observation}

\begin{proof}
If $(v,w)$ is crossed by an edge $(x,y) \in E(H)$, then $(x,w)$ and $(y,w)$ are uncrossed edges (with weights $|\kappa_{xw}|$ and $|\kappa_{yw}|$ respectively), so $x$ and $y$ are non-crossing adjacent to $\mathcal{R}(H)$. 
\end{proof}

\edit{
\cref{obs: st subgraph weighted} will be useful for us later to show that some graphs that we construct in the proof are $(s,t)$-subgraphs.  

\begin{observation}\label{obs: st subgraph weighted}
  Let $H$ be a guarded subgraph of $G$. For any pair of vertices $a,b \in \mathcal{R}(H)$, if $\kappa_{ab} \in E(G)$, then $\kappa_{ab} \in E(\mathcal{R}(H))$.  
\end{observation}

\begin{proof}
Since $\kappa_{ab}$ is an uncrossed edge of $G$, and the robber territory is a maximal subgraph (\cref{def: robber territory}), $\kappa_{ab} \in E(G)$ implies that $\kappa_{ab} \in E(\mathcal{R}(H))$.
\end{proof}
}

Let $\mathcal{U}_1$, $\mathcal{U}_2$ and $\mathcal{U}_3$ be three sets of $|\mathcal{U}|$ cops. The first iteration begins by selecting two arbitrary vertices $s$ and $t$, and letting $\mathcal{U}_1$ guard a shortest non-self-crossing $(s,t)$-path $P \subseteq G$ (\cref{lem: shortest paths do not cross}). This marks the end of the first iteration. Inductively, suppose that at the end of some iteration, a path $P$ or a cycle $C$ has been guarded. For the next iteration, the choice of whether to guard a path or cycle is made by looking at how many vertices of $P$ or $C$ are non-crossing adjacent to the robber territory. Therefore, we consider different cases. (The case analysis here is very similar to that in \cite{tcs_1planar} for 1-planar graphs, which is itself based upon the case analysis used in \cite{bonato_book} for planar graphs.) \edit{For each case, we will implicitly rely on \cref{obs: st subgraph weighted} to conclude that certain graphs are $(s,t)$-subgraphs, and \cref{lem: shortest paths do not cross} to find non-self-crossing shortest $(s,t)$-paths in these $(s,t)$-subgraphs.}

\subparagraph{Guarding a path $P$.} Let path $P$ be guarded by $\mathcal{U}_1$. We have two cases depending upon how many vertices of $P$ are non-crossing adjacent to $\mathcal{R}(P)$.

\medskip
\textbf{Case $(P1)$:} If a single vertex $s \in P$ is non-crossing adjacent to $\mathcal{R}(P)$, then pick any vertex $t$ in $\mathcal{R}(P)$, and use $\mathcal{U}_2$ to guard a shortest non-self-crossing $(s,t)$-path in the $(s,t)$-subgraph $\mathcal{R}(P) \cup S_P(s)$. By \cref{obs: adjacency}, vertex $s$ must be the only vertex of $P$ adjacent (crossing or non-crossing) to $\mathcal{R}(P)$. Hence, $\mathcal{U}_1$ can be freed.

\medskip
\textbf{Case $(P2)$:} Suppose that more than one vertex of $P$ is non-crossing adjacent to $\mathcal{R}(P)$. Then, enumerate the vertices of $P$ based on their distance from one of the end vertices of $P$. Let $s$ and $t$ be the first and last vertices of $P$ that are non-crossing adjacent to $\mathcal{R}(P)$. Use $\mathcal{U}_2$ to guard a shortest non-self-crossing $(s,t)$-path $Q$ in the $(s,t)$-subgraph $\mathcal{R}(P) \cup S_P(s) \cup S_P(t)$. Since no edge of $\mathcal{R}(P)$ crosses an edge of $P$ (\cref{obs: robber territory}), and no edge of $S_P(s) \cup S_P(t)$ crosses $P$ (\cref{def: non-crossing adjacency}), the cycle $C = P(s \dots t) \cup Q$ is non-self-crossing. We let $\mathcal{U}_3$ guard $P(s \dots t)$. If the robber is inside $C$, then clearly $\mathcal{U}_1$ can be freed. If the robber is outside $C$, then the robber cannot visit any vertex of $P \setminus P(s \dots t)$; this is because all non-crossing adjacent vertices of $P$ belong to $P(s \dots t)$, and there can be no crossing-adjacent vertex in $P \setminus P(s \dots t)$ (\cref{obs: adjacency}). Therefore, $\mathcal{U}_1$ can be freed.

\subparagraph{Guarding a cycle $C$.} Let $C = P_1 \cup P_2$ be a non-self-crossing cycle of $G$ guarded by two sets of cops, say $\mathcal{U}_1$ and $\mathcal{U}_2$. Since $C$ is non-self-crossing, we can talk of the sides of $C$, and assume, without loss of generality, that the robber is on the inside of $C$. As for paths, we have different cases depending upon the number of vertices of $C$ that are non-crossing adjacent to $\mathcal{R}(C)$.

\medskip
\textbf{Case $(C1)$:} When there is a single vertex of $C$ that is non-crossing adjacent to $\mathcal{R}(C)$, we do exactly as in Case $(P1)$.

\medskip
\textbf{Case $(C2)$:} Suppose that a single vertex $s \in P_1$ and a single vertex of $t \in P_2$ are non-crossing adjacent to $\mathcal{R}(C)$. Then use $\mathcal{U}_3$ to guard a shortest non-self-crossing $(s,t)$-path in the $(s,t)$-subgraph $\mathcal{R}(C) \cup S_C(s) \cup S_C(t)$. By the case assumption, $s$ and $t$ must be non-adjacent in $C$. Hence, by \cref{obs: adjacency}, these are the only vertices (crossing or non-crossing) adjacent to $\mathcal{R}(C)$. Therefore, the cops $\mathcal{U}_1$ and $\mathcal{U}_2$ can be freed.

\medskip
\textbf{Case $(C3)$:} Up to symmetry, assume that more than one vertex of $P_1$ is non-crossing adjacent to $\mathcal{R}(C)$. Enumerate the vertices of $P_1$ based on their distance from one of the end vertices of $P_1$. Let $s$ and $t$ be the first and last vertices that are non-crossing adjacent to $\mathcal{R}(C)$. Let $P_3$ be a shortest non-self-crossing $(s,t)$-path in the $(s,t)$-subgraph $\mathcal{R}(C) \cup S_C(s) \cup S_C(t)$, and use the cops of $\mathcal{U}_3$ to guard $P_3$. 

Let $x$ and $y$ be the first and last vertices of $P_1$. Let $P_2^+ := P_1(x \dots s) \cup P_2 \cup P_1(t \dots y)$ and $P_1^- := P_1(s \dots t)$. Since no edge of $\mathcal{R}(C)$ crosses an edge of $C$ (\cref{obs: robber territory}) and no edge of $S_C(s) \cup S_C(t)$ crosses $C$ (\cref{def: non-crossing adjacency}), $C_L = P_1^- \cup P_3$ and $C_R = P_3 \cup P_2^+$ are both non-self-crossing cycles. Therefore, the robber must be in the interior of the regions enclosed by one of the two cycles. If the robber is in the interior of $C_L$, then $\mathcal{U}_2$ can be freed because $C_L$ is guarded by $\mathcal{U}_1$ and $\mathcal{U}_3$. Suppose that the robber is in the interior of $C_R$. Since all non-crossing adjacent vertices of $P_1$ belong to $P_1(s \dots t)$, any vertex of $P_1(x \dots s) \cup P_1(t \dots y)$ that is adjacent to $\mathcal{R}(C_R)$ must be crossing-adjacent, and in particular, cross an edge of $P_2$. This implies that $P_2^+$ is guarded by $\mathcal{U}_2$. Therefore, $C_R$ is guarded by $\mathcal{U}_3$ and $\mathcal{U}_2$, and $\mathcal{U}_1$ can be freed.

\medskip
For all the cases above, we show that there is a progression in the number of guarded vertices of $G$.

\begin{claim}\label{obs: increase in size of guarded territory}
At the end of each iteration, the total number of vertices guarded increases by at least one.
\end{claim}

\begin{claimproof}
It is sufficient to show that in all cases above, at least one vertex of the robber territory belongs to the $(s,t)$-path chosen. For Cases $(P1)$ and $(C1)$, this is easy to see because $t$ belongs to the robber territory. For the remaining cases, the $(s,t)$-subgraphs are constructed such that for $H \in \{P,C\}$, the set of edges in $S_H(s)$ and $S_H(t)$ do not include the edge $(s,t)$ (\cref{def: non-crossing adjacency}). Therefore, the $(s,t)$-paths in these cases include some vertex of the robber territory. 
\end{claimproof}

From \cref{obs: increase in size of guarded territory}, we can conclude that after a finite number of iterations, the robber is captured. This proves \cref{lemma: 3 times U_d}. 
\end{proof}

\cref{lem: shortest path guarding,lemma: 3 times U_d} together establish \cref{thm: cop number 1-planar weighted kite edge}. 

\subsection{Implications of \texorpdfstring{\cref{thm: cop number 1-planar weighted kite edge}}{Theorem}.}

In this section, we discuss the implications of \cref{thm: cop number 1-planar weighted kite edge} for 1-plane graphs, $k$-plane graphs and general graphs drawn on the sphere. A summary of the results in this section appears in \cref{table:results}. 
To keep notation compact, we let $c_{f(X)}(G)$ and $c_{f(x)}(G)$ be shortcuts to $c_{f(X(G))}(G)$ and $c_{f(x(G))}(G)$, respectively. 

\subsubsection{1-Plane graphs}

One of the first implications \cref{thm: cop number 1-planar weighted kite edge}  for 1-plane graphs is that $c_{X-1}(G)$ and $c_{x-1}(G)$ are at most linear in $X(G)$ and $x(G)$. 

\begin{corollary}\label{cor: 1-planar}
    If $G$ is a 1-plane graph, then $c_{X - 1}(G) \leq 6X-3$ and $c_{x - 1}(G) \leq 6x+9$.
\end{corollary}

\begin{proof}
Consider the graph $G^\boxtimes$ obtained by kite-augmenting $G$. From \cref{thm: cop number 1-planar weighted kite edge}, setting $\alpha = 1$, we have $c_{X-1}(G^\boxtimes) \leq 3\cdot(2(X(G^\boxtimes)-1) + 1) = 6X(G^\boxtimes)-3$ and $c_{x-1}(G^\boxtimes) \leq 3\cdot(2(x(G^\boxtimes)+1) + 1) = 6x(G^\boxtimes)+9$. As $X(G) = X(G^\boxtimes)$ and $x(G) = x(G^\boxtimes)$, we get the stated result.
\end{proof}

\cref{cor: 1-planar} corroborates, improves and generalises all existing results on cop-numbers of 1-planar graphs. For a maximal 1-planar graph, it was shown in \cite{optimal1plane} that $c(G) \leq 3$. In a maximal 1-planar graph, the endpoints of every crossing induce a $K_4$. Setting $X(G) = 1$ in \cref{cor: 1-planar} gives us the same cop-number. In fact, it shows that 3 cops are sufficient for the larger class of full 1-planar graphs: 1-planar graphs where the endpoints of every crossing induce a $K_4$. (We shall give a much simpler proof of this result in \cref{subsec: map graphs}.) Another class of 1-planar graphs that has been studied is 1-planar graphs embeddable without $\times$-crossings (crossings whose endpoints induce a matching). In \cite{tcs_1planar}, it was shown that $c(G) \leq 21$ for all such 1-planar graphs $G$. \cref{cor: 1-planar} gives us a better bound: as $x(G) = 1$ for 1-planar graphs without $\times$-crossings, we get $c(G) \leq 15$.

\subsubsection{\texorpdfstring{$k$}{k}-Plane graphs}

We now look at $k$-plane graphs. In \cref{cor: k-planar}, we look at distance $X+1$ and $x+2$ cop-numbers obtained by setting $\alpha = 1$ in \cref{thm: 1-planarisation}.

\begin{corollary}\label{cor: k-planar}
 If $G$ is a $k$-plane graph, then $c_{X + 1}(G) \leq 6k(X + 1)-3$ and $c_{x + 2}(G) \leq 6k(x + 2) - 3$.
\end{corollary}

\begin{proof}
From \cref{thm: 1-planarisation}, setting $\alpha=1$, we have $c_{X + 1}(G) \leq c_{X-1}(G^\boxtimes)$ and $c_{x + 2}(G) \leq c_{x-1}(G^\boxtimes)$. From \cref{lemma: bounds on x and X}, we get $X(G^\boxtimes) \leq k(X(G)+1)$ and $x(G^\boxtimes) \leq k(x(G)+2)-2$. From \cref{cor: 1-planar}, we have $c_{X-1}(G^\boxtimes) \leq 6X(G^\boxtimes)-3 \leq 6k(X(G) + 1)-3$ and $c_{x-1}(G^\boxtimes) \leq 6x(G^\boxtimes) + 9 \leq 6k(x + 2) -3$. Hence, the result.
\end{proof}

From \cref{cor: k-planar}, we see that if $G$ is a full $k$-plane graph ($X(G) = 1$), then $c_2(G) \leq 12k-3$. On the other hand, if $G$ is a $k$-plane graph without $\times$-crossings ($x(G) = 1$), then $c_3(G) \leq 18k-3$.

\medskip
A graph $G$ is a \textit{$k$-framed graph} if it has a drawing on the sphere such that its skeleton $\mathrm{sk}(G)$ (the subgraph induced by the set of all uncrossed edges) is simple, biconnected, spans all vertices and each face boundary has at most $k$ edges
\cite{bekos_dframe}. Clearly, a simple $k$-framed graph can have at most $k^2$ edges inside each face. One can draw these graphs such that any two edges cross at most once, hence they are $k^2$-planar graphs.

\begin{corollary}\label{cor: d-framed}
If $G$ is a $k$-framed graph, then $c_{\left \lceil \frac{k+8}{3} \right \rceil}(G) \leq 21$.  
\end{corollary}

\begin{proof}
Since every face of $\mathrm{sk}(G)$ has at most $k$ edges, for every crossing of $G$, there is a path of length at most $k/4$ that connects some pair of consecutive endpoints; therefore, $x(G) \leq k/4$. By setting $\alpha = 4/3$ for $k = x(G)$ in \cref{thm: cop number 1-planar weighted kite edge}, we get $c_{\left \lceil \frac{k+8}{3} \right \rceil}(G) = c_{\left \lceil \frac{4}{3}(\frac{k}{4} + 2) \right \rceil}(G) \leq c_{\left \lceil \frac{4}{3}x \right \rceil - 1}(G^\boxtimes) \leq 21$.
\end{proof}

\subsubsection{General graphs}

We now consider general graphs drawn on the sphere. In \cref{cor: arbitrary x and X = 1}, we look at graphs where all crossings are full and graphs without $\times$-crossings.

\begin{corollary}\label{cor: arbitrary x and X = 1}
For any graph $G$, if $X(G) = 1$, then $c_3(G) \leq 9$, and if $x(G) = 1$, then $c_4(G) \leq 21$.
\end{corollary}

\begin{proof}
From \cref{thm: 1-planarisation}, we have $c_{\ceil{\alpha(X + 1)}}(G) \leq c_{\ceil{\alpha X}-1}(G^\boxtimes)$ and $c_{\ceil{\alpha(x+2)}}(G) \leq c_{\ceil{\alpha x}-1}(G^\boxtimes)$. By setting $\alpha = 3/2$ for $d = X(G^\boxtimes)$ in \cref{thm: cop number 1-planar weighted kite edge}, we get $\beta = \left \lceil \frac{1}{2(\alpha-1)} \right \rceil = 1$, and setting $\alpha = 4/3$ for $d = x(G^\boxtimes)$, we get $\beta = \left \lceil \frac{1}{2(\alpha-1)} \right \rceil+1 = 3$. This implies that $c_3(G) = c_{\ceil{\alpha(X + 1)}}(G) \leq c_{\ceil{\alpha X}-1}(G^\boxtimes) \leq 3$ when $X(G) = 1$, and $c_4(G) = c_{\ceil{\alpha(x + 2)}}(G) \leq c_{\ceil{\alpha x}-1}(G^\boxtimes) \leq 21$ when $x(G) = 1$. 
\end{proof}

From \cref{thm: cop number 1-planar weighted kite edge}, it is easy to see that for $\alpha > 1$, distance $\ceil{\alpha x}-1$ and $\ceil{\alpha X}-1$ cop-numbers of 1-plane graphs are within a constant factor of $1/(\alpha - 1)$. When combined with \cref{thm: 1-planarisation}, one can obtain a similar bound for $c_{\ceil{\alpha(X+1)}}(G)$ and $c_{\ceil{\alpha(x+2)}}(G)$ for arbitrary graphs $G$.   

\begin{corollary}\label{cor: estimates}
For any graph $G$,
\begin{equation*}
c_{\ceil{\alpha(X+1)}}(G) \leq \begin{cases}
             \frac{8}{\alpha - 1}  & \text{if } 1 < \alpha < \frac{3}{2} \\
             9   & \text{if } \alpha \geq \frac{3}{2}
       \end{cases} \qquad
c_{\ceil{\alpha(x+2)}}(G) \leq \begin{cases}
             \frac{11}{\alpha - 1}  & \text{if } 1 < \alpha < \frac{3}{2} \\
             15  & \text{if } \alpha \geq \frac{3}{2}
       \end{cases}
\end{equation*}
\end{corollary}

\begin{proof}
From \cref{thm: 1-planarisation}, we have $c_{\ceil{\alpha(X+1)}}(G) \leq c_{\ceil{\alpha X} - 1}(G^\boxtimes)$ and $c_{\ceil{\alpha(x+2)}}(G) \leq c_{\ceil{\alpha x} - 1}(G^\boxtimes)$. Setting $d := X(G^\boxtimes)$ and $\alpha > 1$ in \cref{thm: cop number 1-planar weighted kite edge}, we get: 
\begin{align*}
(\alpha - 1) \cdot c_{\ceil{\alpha X} - 1}(G^\boxtimes)  &\leq (\alpha-1)\cdot 3\left(2\beta + 1\right)\\ 
&\leq (\alpha-1)\cdot\left(3\left (2\left \lceil \frac{1}{2(\alpha -1)} \right \rceil + 1\right)\right) \\
&\leq (\alpha-1)\cdot\left(6 \left(\frac{1}{2(\alpha -1)} + 1\right )+ 3\right) \\
&\leq (\alpha-1)\cdot \left(\frac{3}{\alpha-1} + 9 \right) \\
&= 3 + 9(\alpha - 1)
\end{align*}
When $\alpha < 3/2$, we have $3 + 9(\alpha - 1) \leq 8$. Therefore, $c_{\ceil{\alpha(X+1)}}(G) \leq c_{\ceil{\alpha X} - 1}(G^\boxtimes) \leq 8/(\alpha - 1)$. Similarly, one can show that $(\alpha - 1) \cdot c_{\ceil{\alpha x} - 1}(G^\boxtimes) \leq 3 + 15(\alpha - 1) \leq 11$ for all $\alpha < 3/2$; we leave the details of this to the reader. On the other hand, $\left \lceil \frac{1}{2(\alpha -1)} \right \rceil = 1$ when $\alpha \geq 3/2$; hence $\beta = 1$ if $d = X(G^\boxtimes)$, and $\beta = 2$ if $d = x(G^\boxtimes)$. This implies that $c_{\ceil{\alpha(X+1)}}(G) \leq c_{\ceil{\alpha x} - 1}(G^\boxtimes) \leq 9$ and $c_{\ceil{\alpha(x+2)}}(G) \leq c_{\ceil{\alpha x} - 1}(G^\boxtimes) \leq 15$ for all $\alpha \geq 3/2$.
\end{proof}

\cref{cor: estimates} shows that $c_{\ceil{\alpha(X+1)}}(G)$ is bounded for all constants $\alpha > 1$. In contrast, the class of diameter 2 graphs are such that $X(G) \leq 2$, but this class of graphs is cop-unbounded \cite{distance_cops_robbers}. In other words, for every integer $m > 0$, there is a graph $G$ such that $X(G) \in \{1,2\}$, but $c(G) \geq m$. In \cref{thm: x/6 cop number}, we show a similar result for arbitrarily large values of $X(G)$.

\begin{theorem}\label{thm: x/6 cop number}
    For every pair of integers $m \geq 1$ and $\mu \geq 6$, there exists a graph $G$ with $X(G) \geq \mu$ and $c_{\floor{X/6}-1}(G) \geq m$.
\end{theorem}

\begin{proof}
Let $H$ be a graph of diameter 2 such that $c(H) \geq m$. To get the distance $d$ cop-number to be at least $m$, we use \cref{thm: distance cop-number 2d+1}, which guarantees that $c(H) \leq c_d(H^{(2d+1)})$. By this, $c(H) \leq c_{(d-1)/2}(H^{(d)})$ when $d$ is odd, and $c(H) \leq c_{d/2 - 1}(H^{(d-1)})$ when $d$ is even. Substituting $d: = \mu+1$, we get: 

\medskip
\noindent \textbf{Case (a): $\mu+1$ is odd.} Set $G := H^{(\mu+1)}$, and we have $c_{\floor{\mu/2}}(G) \geq c(H)$.

\smallskip
\noindent \textbf{Case (b): $\mu+1$ is even.} Set $G := H^{(\mu)}$, and we have $c_{\floor{(\mu-1)/2}}(G) \geq c(H)$.

\medskip
As $H$ has diameter 2, in any drawing of $H$ on the sphere, $1 \leq X(H) \leq 2$. By \cref{lemma: bounds on x and X}, we have $\mu \leq \mu X(H) \leq X(H^{(\mu)}) \leq \mu X(H) + \mu \leq 3\mu$. Likewise, we can show $\mu+1 \leq X(H^{(\mu+1)}) \leq 3(\mu+1)$. As $G = H^{(\mu)}$ or $G = H^{(\mu+1)}$, we have $\mu \leq X(G) \leq 3(\mu+1)$. We now are left with showing that $c_{\floor{X(G)/6}-1}(G) \geq m$. Note that $\floor{X(G)/6}-1 =\floor{\frac{X(G)-6}{6}} \leq \floor{\frac{3(\mu-1)}{6}} = \floor{\frac{\mu-1}{2}} \leq \floor{\frac{\mu}{2}}$. Hence, for all $\mu \geq 6$, we have $X(G) \geq \mu \geq 6$, and $c_{\floor{X/6}-1}(G) \geq c_{\floor{(\mu-1)/2}}(G) \geq c_{\floor{\mu/2}}(G) \geq c(H) \geq m$.
\end{proof}

\section{Graphs on Surfaces}\label{sec: graphs on surfaces}

In this section, we study cop-numbers for $(\Sigma,1)$-planar drawings of graphs on both orientable and non-orientable surfaces. (We assume that the reader is familiar with basic concepts related to graphs on surfaces; see \cite{gross_tucker} for a reference.) All drawings are assumed to be \textit{cellular}; i.e, if $\gamma: G \mapsto S$ is a function that maps a graph $G$ onto a surface $S$, then $S \setminus \gamma(G)$ is a disjoint union of disks. Equivalently, $G$ is cellularly drawn if and only if $G^\times$ (the graph obtained by adding dummy vertices at crossing points) is a cellular embedding on $S$. For orientable surfaces, we use the term \textit{genus} to refer to the number of handles on the surface, whereas for non-orientable surfaces, \textit{genus} refers to the number of crosscaps on the surface.

\subsection{\texorpdfstring{$(\Sigma,1)$}{Sigma1}-Planar Drawings on Orientable and Non-orientable Surfaces.}

In \cref{thm: 1-planar genus}, we extend the results of \cref{thm: cop number 1-planar weighted kite edge} to $(\Sigma,1)$-planar drawings on orientable surfaces. The main ideas for the proof of \cref{thm: 1-planar genus} are inspired from \cite{quilliot1985_genus}, which shows that any graph embeddable (without crossings) on an orientable surface of genus $g$ has cop-number at most $2g + 3$. 

\begin{theorem}\label{thm: 1-planar genus}
Let $G$ be any graph that has a kite-augmented $(\Sigma,1)$-planar graph drawing on an orientable surface $\Sigma$ of genus $g$. For any $\alpha \geq 1$, we have $c_{\ceil{\alpha x} - 1}(G) \leq (2g + 3) \cdot(2 \beta + 1)$ and $c_{\ceil{\alpha X} - 1}(G) \leq (2g+3) \cdot (2\beta + 1) + g$ where 
\begin{equation*}
\beta = \begin{cases}
             d + 1  & \text{if } d = x(G) \text{ and } \alpha = 1 \\
             \left \lceil \frac{1}{2(\alpha - 1)}\right \rceil + 1   & \text{if } d = x(G) \text{ and } \alpha > 1
       \end{cases} \qquad
\beta = \begin{cases}
             d - 1  & \text{if } d = X(G) \text{ and } \alpha = 1 \\
             \left \lceil \frac{1}{2(\alpha - 1)} \right \rceil  & \text{if } d = X(G) \text{ and } \alpha > 1
       \end{cases}
\end{equation*}
\end{theorem}

\begin{proof}
Our main idea is to guard a shortest non-contractible cycle in $G$, then delete the cycle to get a graph that is embeddable on a surface of smaller genus, and use induction to complete the proof. The base case $g = 0$ was already proven in \cref{thm: cop number 1-planar weighted kite edge}, so assume that $g > 0$.  We first show that there exists a non-contractible cycle in $G$. As $G^\times$ is a cellular embedding, there exists a non-contractible cycle of $G^\times$ in $S$. This can be turned into a non-contractible cycle of $G$: for every pair of edges $\{(u,d), (d,v)\}$ on the cycle where $u$ and $v$ are consecutive endpoints of a crossing at $d$, replace the two edges with $\kappa_{uv}$. With this, we now get a non-contractible cycle of $G$. 

Since $G$ has a non-contractible cycle, there exists a non-empty shortest non-contractible cycle $C \subseteq G$. (As $G$ is drawn with crossings, the edges of $C$ may not necessarily trace a simple closed curve on the surface.) An especially useful property of shortest non-contractible cycles is that every pair of vertices on the cycle can be connected by a shortest path that is a subset of the cycle (see also \cite{thomassen_non_contractible}). 

\begin{claim}\label{claim: non-contr cycle}
   For any pair of vertices $x,y \in V(C)$, the shortest $(x,y)$-path on $C$ is also a shortest $(x,y)$-path in $G$. 
\end{claim}

\begin{claimproof}
   Let $P_1$ and $P_2$ be the two $(x,y)$-paths on $C$. Suppose, for contradiction, that there is a path $P_3$ in $G$ such that $|P_3| < \min\{|P_1|, |P_2|\}$. Then the two cycles $C_{13} := P_1 \cup P_3$ and $C_{23} := P_2 \cup P_3$ must be contractible, since they are both shorter than $C$, which is a shortest non-contractible cycle of $G$. This would imply that the concatenation $C_{13} \cdot C_{23} = (P_1(x \dots y) \cdot P_3(y \dots x)) \cdot (P_3(x \dots y) \cdot P_2(y \dots x)) = P_1(x \dots y) \cdot P_2(y \dots x) = C$ is contractible; this is a contradiction.  
\end{claimproof}

\cref{claim: non-contr cycle} implies that one can choose shortest paths $P_1$ and $P_2$ with $|P_1| = |P_2| = \floor{|C|/2}$ such that $C = P_1 \cup P_2$, if $|C|$ is even, and $C = P_1 \cup P_2 \cup \{e\}$, if $|C|$ is odd. As $C$ can be expressed as the union of shortest paths, it can be crossing-guarded using few cops. We discuss this in detail below, considering the differences in the parity of $|C|$.

\medskip
\textbf{Case (a): $|C|$ is even.} By \cref{lem: shortest path guarding}, each of $P_1$ and $P_2$ can be crossing-guarded at distance $\ceil{\alpha d}-1$ using $2\beta + 1$ cops. Therefore, $C$ can be crossing guarded at distance $\ceil{\alpha d}-1$ using $2\cdot(2\beta + 1)$ cops.

\medskip
\textbf{Case (b): $|C|$ is odd.} As with Case (a), both $P_1$ and $P_2$ can be crossing-guarded at distance $\ceil{\alpha d}-1$ using $2\beta + 1$ cops. For $C= P_1 \cup P_2 \cup \{e\}$ to be crossing-guarded, we need to ensure that $e$ is crossing-guarded too. Let the crossing at $e = (u,v)$ be $\{(u,v), (w,x)\}$ where $u \in P_1$ and $v \in P_2$. If $d = X(G)$, we place a single cop at $w$ so that $e$ is never crossed. If $d = x(G)$, we will show that this additional cop can be avoided. 

Let $d = x(G)$. Suppose that the robber lands on the vertex $w$ in an attempt to cross $e$ and land on $x$ in the next round. As in the proof of \cref{lem: shortest path guarding}, we can show that $|\kappa_{wu}| > d$ and $|\kappa_{wv}| > d$ since the range of the cops guarding $P_1$ and $P_2$ is at least $d+1$ (\cref{obs: range and d}). Hence, one can assume, up to symmetry, that $|\kappa_{xu}| \leq d$. As this implies that $|\kappa_{wu}| \leq d+1$, after the cops' movement on $P_1$ in response to the robber landing on $w$, vertex $u$ comes inside the range of the cops. If the robber now moves to $x$, he can be captured by a cop at distance $\ceil{\alpha d}-1$ (\cref{obs: range of P}).      

\medskip
For both cases above, $C$ can be crossing-guarded at distance $\ceil{\alpha d} - 1$ using $2 \cdot(2\beta +1)$ cops when $d = x(G)$ and using $2 \cdot (2\beta + 1) + 1$ cops when $d = X(G)$. Now, delete all vertices of $C$ and all edges that cross some edge of $C$. The resulting drawing is not cellular, since otherwise $C$ can be drawn inside a disc, implying that $C$ is contractible---a contradiction. Let $H$ be the connected component of the resulting graph that contains the robber. Note that $H$ is a kite-augmented $(\Sigma,1)$-planar drawing, as any crossing of $G$ with some endpoint incident with $C$ can no longer be a crossing of $H$. Since $H^\times$ is not a cellular embedding, one can re-embed $H^\times$ on an orientable surface $\Sigma'$ of smaller genus with the same rotation system (see Section 3.4 of \cite{mohar2001graphs}). Since the rotation system is preserved, all pairs of crossing edges of $H$ are preserved in the new cellular embedding. As $H$ is a $(\Sigma',1)$-planar kite-augmented drawing on a surface of smaller genus, we can do induction on the genus $g$, and obtain the desired result.   
\end{proof}

\edit{

Combining \cref{thm: 1-planar genus} with \cref{thm: 1-planarisation}, we get the following corollary.

\begin{corollary}
Let $G$ be any graph drawn on an orientable surface $\Sigma$ of genus $g$. For any $\alpha \geq 1$, we have $c_{\ceil{\alpha (x + 2)}}(G) \leq (2g + 3) \cdot(2 \beta + 1)$ and $c_{\ceil{\alpha (X + 1)}}(G) \leq (2g+3) \cdot (2\beta + 1) + g$ where 
\begin{equation*}
\beta = \begin{cases}
             d + 1  & \text{if } d = x(G) \text{ and } \alpha = 1 \\
             \left \lceil \frac{1}{2(\alpha - 1)}\right \rceil + 1   & \text{if } d = x(G) \text{ and } \alpha > 1
       \end{cases} \qquad
\beta = \begin{cases}
             d - 1  & \text{if } d = X(G) \text{ and } \alpha = 1 \\
             \left \lceil \frac{1}{2(\alpha - 1)} \right \rceil  & \text{if } d = X(G) \text{ and } \alpha > 1
       \end{cases}
\end{equation*}
\end{corollary}
}

We now shift our focus to non-orientable surfaces. In \cite{joret_non_orientable}, it was shown that for any graph $G$ embedded (without crossings) on a non-orientable surface of genus $g$, there is a graph $H$ embedded (without crossings) on an orientable surface of genus $g-1$ such that $c(G) \leq c(H)$. The main idea for this stems from the well-known fact that every non-orientable surface $\Sigma_{g}$ of genus $g$ has an orientable surface $\Sigma_{g-1}$ of genus $g-1$ as a \textit{2-sheeted covering space}. In other words, there is a continuous function $\pi:  \Sigma_{g-1} \mapsto \Sigma_{g}$ such that for every point $x \in \Sigma_{g}$, there exist an open neighbourhood $U_x$ containing $x$ such that $\pi^{-1}(U_x)$ is a disjoint union of two open sets $V_x^1$ and $V_x^2$ where $\pi(V_x^i)$ is homeomorphic to $U_x$ for $i \in \{1,2\}$. So, if a graph $G$ embedded on $\Sigma_{g}$ is lifted into a graph $H$ in $\Sigma_{g-1}$, then there are two copies for every vertex, edge and face of $G$ in $H$ (a simple application of Euler's formula then explains the genus $g-1$ for the covering space). Here, we use ideas from \cite{joret_non_orientable} to derive a similar result for kite-augmented 1-planar drawings on non-orientable surfaces.

\begin{theorem}\label{thm: non-orientable}
Let $G$ be a graph drawn on a non-orientable surface $\Sigma_g$ of genus $g$. Then there exist a $(\Sigma_g,1)$-plane graph $G^\boxtimes$ and a $(\Sigma_{g-1},1)$-plane graph $H$, where $\Sigma_{g-1}$ is an orientable surface of genus $g-1$, such that $X(G^\boxtimes) = X(H)$, $x(G^\boxtimes) = x(H)$, and $c_d(G^\boxtimes) \leq c_d(H) \leq 2 \cdot c_d(G^\boxtimes)$ for any integer $d \geq 0$.
\end{theorem}

\begin{proof}
Let $G^\boxtimes$ be the kite-augmentation (after $(\Sigma_g,1)$-planarisation) where for every crossing at a point $c$, and every consecutive pair of endpoints $u,v$, the path $\kappa_{uv}$ is drawn close to the curve $u\textrm{---}c\textrm{---}v$. Let $\Sigma_{g-1}$ be the orientable surface of genus $g-1$ that is a 2-sheeted covering space of $\Sigma_{g}$, and let $\pi: \Sigma_{g-1} \mapsto \Sigma_{g}$ be the covering function. Let $H$ be the graph obtained by lifting $G$ from $\Sigma_{g}$ onto $\Sigma_{g-1}$. As each path $\kappa_{uv}$ in $G^\boxtimes$ traces the curve $u\textrm{---}c\textrm{---}v$, every crossing in $H$ contains paths that trace similar curves between consecutive endpoints of the crossing, and these have the same length as in the projection of the crossing on $G^\boxtimes$. Therefore, $X(H) = X(G^\boxtimes)$ and $x(H) = x(G^\boxtimes)$. 

We now prove the bounds on cop-numbers. The main idea is that all cop-moves and robber-moves in either $G^\boxtimes$ or $H$ can be simulated on the other graph using the covering function $\pi$. One can easily show that $c_d(G^\boxtimes) \leq c_d(H)$ by simply re-hashing the proof in \cite{joret_non_orientable} for distances $d > 0$; hence we restrict ourselves to showing that $c_d(H) \leq 2 \cdot c_d(G^\boxtimes)$. For this, we simulate the robber movements of $H$ in $G^\boxtimes$ and the cop movements of $G^\boxtimes$ in $H$. If initially a cop is placed on a vertex $x \in V(G^\boxtimes)$, then we place two cops on the two vertices of $\pi^{-1}(x)$ in $H$. Likewise, if the robber is initially placed on a vertex $y \in V(H)$, we place the robber on $\pi(y) \in V(G^\boxtimes)$. Note that for every edge $(x,y) \in E(G^\boxtimes)$, there are two edges connecting vertices $\pi^{-1}(x)$ and $\pi^{-1}(y)$ in $H$. Therefore, we can simulate a cop move along an edge from $x$ to $y$ through the moves of two cops from $\pi^{-1}(x)$ to $\pi^{-1}(y)$. Likewise, a robber move in $H$ from a vertex $x$ to a vertex $y$ is simulated as the move from $\pi(x)$ to $\pi(y)$ in $G^\boxtimes$. If a cop in $G^\boxtimes$ at vertex $x$ captures the robber at vertex $y$ within distance $d$, then one of the two cops in $\pi^{-1}(x)$ captures the robber in $\pi^{-1}(y)$ within distance $d$. 
\end{proof}

% As $X(G) = X(H)$ and $x(G) = x(H)$ we can use \cref{thm: 1-planar genus} to get the following corollary. 

% \begin{corollary}
% Let $G$ be any graph drawn on a non-orientable surface $\Sigma$ of genus $g$. There exists a graph $G^\boxtimes$ such that for any $\alpha \geq 1$, we have $c_{\ceil{\alpha x} - 1}(G^\boxtimes) \leq (2g + 1) \cdot(2 \beta + 1)$ and $c_{\ceil{\alpha X} - 1}(G^\boxtimes) \leq (2g+1) \cdot (2\beta + 1) + (g-1)$ where 
% \begin{equation*}
% \beta = \begin{cases}
%              d + 1  & \text{if } d = x(G) \text{ and } \alpha = 1 \\
%              \left \lceil \frac{1}{2(\alpha - 1)}\right \rceil + 1   & \text{if } d = x(G) \text{ and } \alpha > 1
%        \end{cases} \qquad
% \beta = \begin{cases}
%              d - 1  & \text{if } d = X(G) \text{ and } \alpha = 1 \\
%              \left \lceil \frac{1}{2(\alpha - 1)} \right \rceil  & \text{if } d = X(G) \text{ and } \alpha > 1
%        \end{cases}
% \end{equation*}
% \end{corollary}

\edit{
\begin{corollary}
Let $G$ be any graph drawn on a non-orientable surface $\Sigma$ of genus $g$. For any $\alpha \geq 1$, we have $c_{\ceil{\alpha (x + 2)}}(G) \leq (2g + 1) \cdot(2 \beta + 1)$ and $c_{\ceil{\alpha (X + 1)}}(G) \leq (2g+1) \cdot (2\beta + 1) + (g-1)$ where 
\begin{equation*}
\beta = \begin{cases}
             d + 1  & \text{if } d = x(G) \text{ and } \alpha = 1 \\
             \left \lceil \frac{1}{2(\alpha - 1)}\right \rceil + 1   & \text{if } d = x(G) \text{ and } \alpha > 1
       \end{cases} \qquad
\beta = \begin{cases}
             d - 1  & \text{if } d = X(G) \text{ and } \alpha = 1 \\
             \left \lceil \frac{1}{2(\alpha - 1)} \right \rceil  & \text{if } d = X(G) \text{ and } \alpha > 1
       \end{cases}
\end{equation*}
\end{corollary}

\begin{proof}
From \cref{thm: non-orientable}, we know that there exist graphs $G^\boxtimes$ and $H$, drawn respectively on $\Sigma$ and an orientable surface of genus $g-1$, such that $c_{\ceil{\alpha x(G^\boxtimes)} -1 }(G^\boxtimes) \leq c_{\ceil{\alpha x(H)} -1 }(H)$ and $c_{\ceil{\alpha X(G^\boxtimes)} -1 }(G^\boxtimes) \leq c_{\ceil{\alpha X(H)} -1 }(H)$. From \cref{thm: 1-planar genus}, we have that $c_{\ceil{\alpha x(H)} -1 }(H) \leq (2g + 1) \cdot(2 \beta + 1)$ and $c_{\ceil{\alpha X(H)} -1 }(H) \leq (2g+1) \cdot (2\beta + 1) + (g-1)$ for the values of $\beta$ as stated above. From \cref{thm: 1-planarisation}, we know that $c_{\ceil{\alpha (x(G) + 2)}}(G) \leq c_{\ceil{\alpha x(G^\boxtimes)} -1 }(G^\boxtimes)$ and $c_{\ceil{\alpha (X(G) + 1)}}(G) \leq c_{\ceil{\alpha X(G^\boxtimes)} -1 }(G^\boxtimes)$ for all graphs $G^\boxtimes$. Therefore, we get the stated result.
\end{proof}
}

\subsection{Map Graphs}\label{subsec: map graphs}

We now consider a special class of graphs on surfaces called map graphs. A graph $G$ is a \textit{map graph on a surface $S$} if the vertex set of $G$ can be represented by internally-disjoint closed-disc homeomorphs called \textit{nations} on $S$, and a pair of vertices are adjacent in $G$ if and only if their corresponding nations touch each other. The above definition generalises \textit{map graphs} originally defined for the sphere \cite{map_graphs} to arbitrary surfaces. In \cite{map_graphs}, it is shown that a graph $G = (V,E)$ is a map graph on the sphere if and only if there is a planar bipartite graph $H$, called the \textit{witness}, with bipartition $V(H) = \{V, F\}$ such that $G = H^2[V]$. That is, two vertices $u$ and $v$ are adjacent in $G$ if and only if there is a path $(u,f,v)$ in $H$ where $u,v \in V$ and $f \in F$. The witness graph $H$ is obtained from the map representation of $G$ as follows: insert a \textit{nation vertex} in the interior of each nation and a \textit{boundary vertex} at points where two or more nations touch each other; then join a nation vertex with a boundary vertex if and only if the boundary point is part of the nation. Conversely, given a plane bipartite witness graph $H$ with bipartition $V(H) = \{V,F\}$, we draw a star-shaped region around each vertex of $V$ such that its boundary contains all vertices of $F$ that the vertex is adjacent to. This gives us a map representation of $G = H^2[V]$. Clearly, this combinatorial characterisation of map graphs in terms of witness graphs extends to map graphs on surfaces. That is, a graph $G = (V,E)$ is a map graph on a surface $S$ if and only if there is a bipartite \textit{witness} graph $H$ embeddable on $S$ (without crossings), with bipartition $V(H) = \{V,F\}$, such that $G = H^2[V]$. 

\begin{theorem}\label{thm: cop number map graph}
If $G$ has a map representation on a surface with a graph $H$ as a witness, then $c(G) \leq c(H)$.
\end{theorem}

\begin{proof}
To prove \cref{thm: cop number map graph}, we employ our usual strategy of simulating the robber-moves of $G$ on $H$, and the corresponding cop-moves of $H$ on $G$. Let $\{U_1', \dots, U_p'\}$ be a set of $p$ cops where $p = c(H)$. We will show that a set $\{U_1, \dots, U_p\}$ of $p$ cops in $G$ is sufficient to capture a robber in $G$. Consider the first round of the game. For every cop $U_i'$ occupying a vertex $v \in V(H)$: if $v \in V(G)$, place $U_i$ on $v$; else place $U_i$ on a vertex of $G$ neighbouring $v$. If the robber's initial position in $G$ is on a vertex $x \in V(G)$, then place the robber on the same vertex $x \in V(H)$. Thereafter, every single round played in $G$ can be simulated as two rounds played in $H$: If the robber in $G$ moves along an edge $(x,y) \in E(G)$, then the robber in $H$ moves along a path $(x,f,y)$ in $H$, for some vertex $f \in V(H) \setminus V(G)$ (this is possible since $G = H^2[V]$). Similarly, the moves of a cop $U_i'$ can be simulated by $U_i$ in $G$ such that at the end of every two rounds played in $H$: if $U_i'$ is on a vertex $v \in V(G)$, so is $U_i$; else if $U_i'$ is on a vertex $f \notin V(G)$, then $U_i$ is on some vertex of $G$ neighbouring $f$.

Now consider the final round in $H$ when the robber gets captured by a cop $U_i'$. Suppose that the robber in $G$ takes an edge $(x,y)$. As seen before, this corresponds to a path $(x,f,y)$ in $H$. If the robber is captured on $y$, then he is captured by $U_i$ in $G$. If the robber is captured on $f$, then $U_i$ captures the robber on $y$ as the neighbours of $f$ induce a clique in $G$.      
\end{proof}

\cref{thm: cop number map graph} gives us a simple proof of the fact that full 1-planar graphs have cop-number at most 3 (\cref{cor: 1-planar}), as all full 1-planar graphs are map graphs on the sphere \cite{Brandenburg_4_map}. Another interesting corollary is that all optimal 2-planar graphs also have cop-number at most 3. This is again due to the fact that all optimal 2-planar graphs are map graphs on the sphere \cite{optimal2plane}.

\begin{corollary}\label{cor: map graphs}
    If $G$ is a full 1-planar or an optimal 2-planar graph, then $c(G) \leq 3$.
\end{corollary}

%
%% Bibliography
%%

%% Please use bibtex, 

% \pagenumbering{Roman}
\section{Open Problems for Future Work}\label{sec: conclusion}

In this paper, we initiated an extensive study of the interplay between distance $d$ cop-numbers of graphs and distances between crossing pairs in the drawing. The tools and techniques that we used in the paper are sufficiently general to handle arbitrary graphs drawn on surfaces. We showed that a small value of $X(G)$ or $x(G)$, for any graph $G$ drawn on a surface of small genus, implies a small distance $d$ cop-number. Despite this, there are still many questions that are elusive; we mention a few of them below.

In \cite{tcs_1planar}, it was shown that if $G$ is a 1-plane graph without $\times$-crossings, then $c(G) \leq 21$. In this paper, we improved the bound to $c(G) \leq 15$ (\cref{cor: arbitrary x and X = 1}). However, we do not know whether the cop-number of a 1-plane graph $G$ with $x(G) = X(G) = 2$ is bounded. It is unlikely that we can drop the requirement of 1-planarity, since graphs of diameter 2 are known to be cop-unbounded. 

\begin{problem}\label{problem: x and X}
If $G$ is a 1-plane graph with $X(G) = x(G) = 2$, then is $c(G)$ bounded?
\end{problem}

If one kite-augments the graph $G$ in \cref{problem: x and X}, then each face of $\mathrm{sk}(G)$ that contains a crossing has 8 edges on its face-boundary. From \cref{thm: single crossing per face}, we know that the cop-number is unbounded for graphs with at most one crossing per face of $\mathrm{sk}(G)$. However, the construction of these graphs are such that they do not have a small number of edges that bound faces of $\mathrm{sk}(G)$.  This leads us to the following question. 

\begin{problem}\label{problem: d-framed}
Does there exist a function $f$ such that if $G$ is a $k$-framed graph, then $c(G) \leq f(k)$?
\end{problem}

In \cref{thm: 1-planar genus}, we gave a bound on the distance $d$ cop-number of graphs drawn on orientable surfaces. Our techniques were based on the ones used in the earliest paper on cop-number for graphs embedded (without crossings) on orientable surfaces \cite{quilliot1985_genus}. Since then, many improvements have been made; see \cite{schroder2001copnumber,genus_currentbest,conference_best_genus} instance. We leave it as an open problem to adapt those techniques to improve the bound in \cref{thm: 1-planar genus}.

\begin{problem}
Improve the bound in \cref{thm: 1-planar genus} for cop-number of a graph drawn on orientable surfaces. 
\end{problem}

\bibliography{References}

\end{document}